\documentclass[english,cleveref,autoref]{lipics-v2019}
\usepackage[T1]{fontenc}
\usepackage[utf8]{inputenc}

\usepackage{comment}
\usepackage{graphicx,amssymb,amsmath}
\usepackage{xcolor}
\usepackage{xspace}
\usepackage{todonotes}
\usepackage{enumerate}
\usepackage[font=small,labelfont=bf]{caption}
\usepackage[subrefformat=parens,labelfont=default,font=small]{subcaption}
\usepackage{enumitem}
\usepackage{booktabs}

\title{Reachability of turn sequences}

\author{William S. Evans}{University of British Columbia, Vancouver, B.C., Canada}{will@cs.ubc.ca}{}{}

\author{Noushin Saeedi}{University of British Columbia, Vancouver, B.C., Canada}{noushins@cs.ubc.ca}{}{}

\author{Chan-Su Shin}{Hankuk University of Foreign Studies, Yongin, Korea}{cssin@hufs.ac.kr}{}{}

\author{Hyun Tark}{Hankuk University of Foreign Studies, Yongin, Korea}{03030303gov@gmail.com}{}{}

\authorrunning{W.~Evans, N. Saeedi, C.-S. Shin and H. Tark}

\Copyright{}



\funding{W. Evans and N. Saeedi were funded by NSERC Discovery grant. C.-S. Shin was supported by the National Research Foundation of Korea(NRF) grant funded by the Korea government(MSIT) (No. 2019R1F1A1058963).}

\graphicspath{{figs/}}

\hideLIPIcs

\newtheorem{observation}{Observation}

\newcommand{\Z}{\ensuremath{{\mathbb Z}}}
\newcommand{\Zo}{\ensuremath{{{\mathbb Z}_{\geq 0}}}}
\newcommand{\lt}{\texttt{L}}
\newcommand{\rt}{\texttt{R}}
\newcommand{\sgn}{\mathrm{sgn}}
\newcommand{\floor}[1]{\lfloor #1 \rfloor}
\newcommand{\ceil}[1]{\lceil #1 \rceil}

\graphicspath{{figs/}}

\begin{document}

\maketitle

\begin{abstract}
A turn sequence of left and right turns is realized as a simple
rectilinear chain of integral segments whose turns at its bends are
the same as the turn sequence. The chain starts from the origin and
ends at some point which we call a \emph{reachable point} of the turn sequence. 
We investigate the combinatorial and geometric properties of
the set of reachable points of a given turn sequence such as the
shape, connectedness, and sufficient and necessary conditions on
the reachability to the four signed axes.
We also prove the upper and lower bounds on the maximum distance
from the origin to the closest reachable point on signed axes for a turn sequence. 
The bounds are expressed in terms of the difference between 
the number of left and right turns in the sequence as well as,
in certain cases, the length of the maximal monotone prefix or suffix 
of the turn sequence.
The bounds are exactly matched or tight within additive constants 
for some signed axes.
\end{abstract}

\section{Introduction}

We consider the problem of characterizing and approximating,
for a given \emph{turn sequence}, 
the set of points reachable by the sequence.
A turn sequence consists of left turns $\lt$ and right turns $\rt$.
For a turn sequence $\sigma = \sigma_1\sigma_2\cdots \sigma_n$
where $\sigma_i \in \{\lt, \rt\}$, a rectilinear chain
\emph{realizes} the sequence $\sigma$ 
if the chain is simple (i.e., has no self-intersection), consists of integral segments, starts from the origin $o$,
makes the sequence $\sigma$ of left turns ($\lt$) and right turns ($\rt$) in order,
and ends at a point $p \in \Z^2$.
We assume that the first segment of such chains is horizontal and 
heads to the east, i.e., 
the first turn occurs at $(t, 0)$ for some positive integer $t$.
The endpoint $p$ is said to be \emph{reachable} by $\sigma$.
Let $A(\sigma)$ be the set of points in $\Z^2$ reachable
by a turn sequence $\sigma$.

One may think of $A(\sigma)$ as those points reachable by a robot
following the sequence $\sigma$ of turn commands when the distance
the robot travels between turns is a positive but arbitrary integer.
Where can the robot end up? 
How close the robots can reach a point from the origin by obeying the turn sequence?

\paragraph*{Related work}
This type of question arose in the study of \emph{turtle graphics}~\cite{P80},
in which a turtle obeys a sequence of commands 
such as ``move forward 10 units'' and ``turn right 90 degrees''.
The trace made by the turtle creates a chain or polygon
on a display device.  The metaphor helped children understand 
basic geometric shapes.

The problem of constructing chains or polygons with
restricted angles
shows up in curvature-constrained motion
planning~\cite{AgarwalRaghavanTamaki95},
angle-restricted tours~\cite{Fekete92PhD,FeketeWoeginger97},
and restricted orientation
geometry~\cite{Rawlins87PhD}. 
Culberson and Rawlins~\cite{cr-socg85} and Hartley~\cite{h-dpgas-IPL89}
studied the problem of realizing a simple polygon by an angle sequence.
An angle sequence is a sequence of the angles at vertices 
(say, in the counterclockwise direction) 
along the boundary of a polygon.
They presented algorithms to construct the polygon 
from a given sequence of exterior angles whose sum is $2\pi$.

A rectilinear angle sequence has only two angles, 
$+90^\circ$ and $-90^\circ$, thus it is the same as the turn sequence
considered in this paper. 
A turn sequence can be realized by a rectilinear simple polygon if and only if
the number of left and right turns in the sequence differs 
by exactly four~\cite{cr-socg85,Sack84PhD}.
Sack~\cite{Sack84PhD} represented a rectilinear polygon 
as a \emph{label-sequence} of integer labels that are defined as 
the difference of the numbers of left and right turns 
at each edge from an arbitrary starting edge,
and also presented a drawing algorithm that realizes a given
label-sequence as a simple rectilinear polygon in linear time. 
Bae et al.~\cite{BaeOS-cgta19} showed tight worst-case bounds 
on the minimum and maximum area of the rectilinear polygon that 
realizes a given turn sequence.
Finding realizations that minimize the area or perimeter 
of the rectilinear polygon or the area of the bounding box of the polygon
is NP-hard,
however, the special case of monotone rectilinear polygons 
can be computed in polynomial time~\cite{efkssw-cgta22,thesisFleszar}. 

A popular variant takes a sequence of angles defined by all vertices
visible from each vertex as input~\cite{cw-cgta12,cgt-isaac13,dmw-cgta11}. 
The goal is to reconstruct a polygon from the information on 
angles and visibility.
Another variant reconstructs a rectilinear polygon from a set of 
points, i.e., coordinates of the vertices, instead of angles,
obtained by laser scanning devices~\cite{bds-tcs11}.

For our problem, determining the reachable region from an angle (turn) sequence,
there appears to be little that is known.
Culberson and Rawlins~\cite{cr-socg85} mention as future work
``spline'' problem: to draw a polygonal curve 
between two given points such that the turning angles of 
the curve form a given angle sequence. 
However, they do not suggest any approach to solving this problem.

\paragraph*{Our contribution}
We first show that the reachable set $A(\sigma)$ 
for a turn sequence $\sigma$ that contains a \emph{hook}, i.e., two
consecutive left turns or right turns, is a union of
at most four halfplanes, not containing the origin, whose bounding
lines are orthogonal to the four signed axes.
If $\sigma$ has no hooks, it is called a \emph{staircase} and has an
easily computed reachable set.
The proof is based on two crucial lemmas: 
the Stretching Lemma and the Axis Lemma
in Section~\ref{sec:cut_stretching}.
The particular halfplanes are determined 
by the orientation of the hooks (Theorem~\ref{thm:reachAxis} in
Section~\ref{sec:reachable_set}).
Using this characterization, 
we prove that both $A(\sigma)$ and its complement,
i.e., the unreachable set $\Z^2 \setminus A(\sigma)$, are connected
(Theorem~\ref{thm:connected} in Section~\ref{sec:reachable_set}).
The boundary lines of the halfplanes forming $A(\sigma)$ 
for any turn sequence $\sigma$ with hooks 
are determined by
the closest reachable points from the origin $o$ on the signed axes.
Thus, it is important to know these points
in order to calculate $A(\sigma)$ accurately.
We give upper and lower bounds on the distance from $o$ 
to the closest points on each signed axis in
Section~\ref{sec:closest}.
The upper bounds rely on an algorithm that provides an approximation
to the reachable region (Theorem~\ref{thm:bounding-x-minus} 
in Section~\ref{subsec:upper_bound_x_minus}
and 
Theorem~\ref{thm:x-plus-conclusion} 
in Section~\ref{subsec:upper_bound_x_plus}).
The lower bounds are derived from the \emph{rotation number}~\cite{gs-tams90} of
polygons, the sum of angle changes between two adjacent edges
(Theorem~\ref{thm:lowerboundconclusion} in Section~\ref{subsec:lower_bound_x_plus}
and Theorem~\ref{thm:lowerboundconclusion2} in Section~\ref{subsubsec:lower_bound_x_minus}).
These bounds are expressed in terms of
the difference between the number of left and right turns in $\sigma$ 
and, for some signed axes and turn sequences with certain patterns, 
the length of the maximal prefix or suffix of the turn sequence that is
\emph{monotone} (i.e., does not contain a certain type of hook) or 
\emph{staircase} (i.e., does not contain a hook).
The bounds for some signed axis are exactly matched or 
tight to within an additive constant.

\section{Cutting and Stretching}
\label{sec:cut_stretching}

Let $p$ be the endpoint of a chain $C$ that realizes a turn sequence
$\sigma$.
We can slightly modify $C$ to show that other points are in the
reachable region $A(\sigma)$.
Our mechanism for doing this relies on stretching 
(or lengthening) a subset of the segments in the chain $C$
that are selected by a cut.
A \emph{cut} of a chain $C$ is an $x$-monotone (or $y$-monotone)
curve extending from $-\infty$ to $+\infty$ in the $x$
(respectively, $y$) dimension that
(1) separates $o$ and $p$ and (2) intersects 
only vertical (respectively, horizontal) segments of $C$.
We \emph{stretch} $C$ using the cut by lengthening every segment in $C$
crossed by the cut by the same (positive) integral amount.
As long as some segment in $C$ crosses the cut, this creates a new
chain $C'$ starting at $o$ that has the same turn sequence as $C$,
does not self-intersect (like $C$), and reaches points in $A(\sigma)$
other than $p$.

\begin{figure}[th]
\centering
    \includegraphics[width=\textwidth]{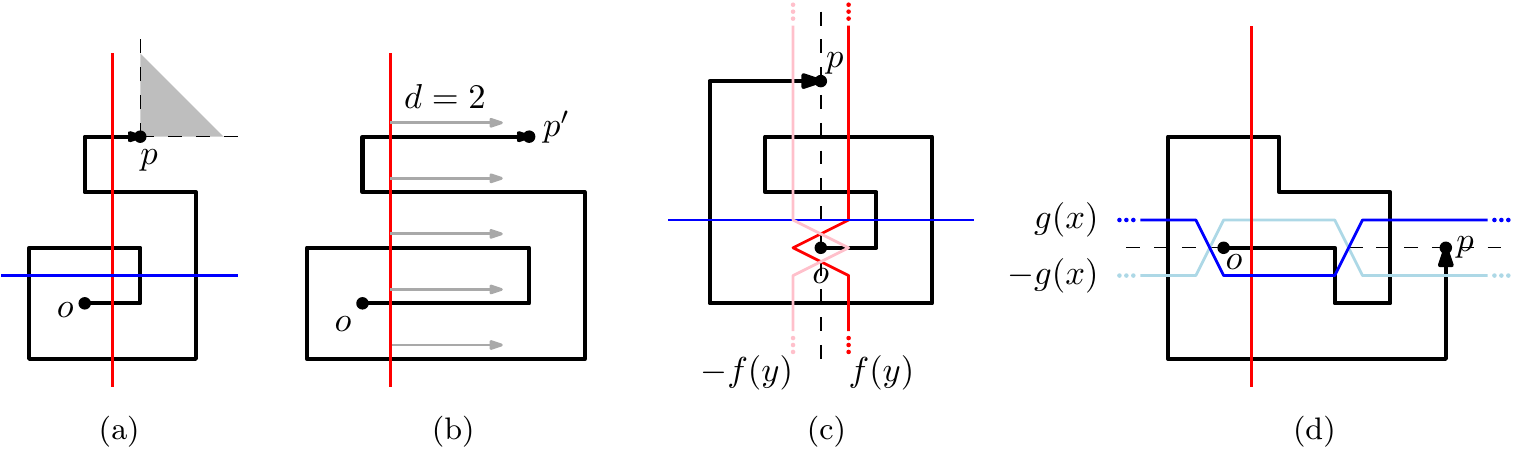}
\caption{
	(a) Vertical cut (red) and horizontal cut (blue). 
	(b) Stretching the chain by $2$ using five horizontal segments
	intersected by a vertical cut.
	(c)-(d) Cuts for the reachable point $p$ on $y$-axis and $x$-axis.}
\label{fig:stretch}
\end{figure}
A vertical cut (the vertical line at $x=1/2$) 
shown in Figure~\ref{fig:stretch}(a)-(b)
intersects five horizontal segments of $C$.
Stretching $C$ (lengthening the five horizontal segments)
by $d >0$ units using this cut creates a chain $C'$ 
with endpoint $p' = p + (d,0)$ that also belongs to $A(\sigma)$.
As a result, we observe that all the points on the horizontal ray
$p + (d,0)$ with integer $d > 0$ are reachable by $\sigma$.
Similarly, using the horizontal cut $y=1/2$ in Figure~\ref{fig:stretch}(a), all the points on the vertical ray $p + (0,d)$ are also
reachable.

We can generalize this stretching procedure to obtain 
the following lemma.
We need some notation for it:
For a non-zero value $a$, $\sgn(a)$ represents its sign, $+1$ or $-1$.
Let $V^+(p)$ and $H^+(p)$ be, respectively, the halfplanes of the vertical and horizontal lines through $p$ that do not contain $o$.
Precisely, if $p = (a, b)$ then
$V^+(p) := \{ (a+\sgn(a)\cdot i,j) \mid i \in \Zo, j \in \Z \}$, and 
$H^+(p) := \{ (i, b+\sgn(b)\cdot j) \mid i\in \Z, j \in \Zo \}$.
Let $Q^+(p) := V^+(p) \cap H^+(p)$ be the quadrant at $p$ that is diagonally
opposite to the quadrant containing $o$, see the shaded region in Figure~\ref{fig:stretch}(a).

\begin{lemma}[Stretching Lemma]
\label{lem:stretch}
	For a turn sequence $\sigma$ with at least one turn,
	\begin{itemize}
	\item[]{\rm (1)} 
		if $p = (a,b) \in A(\sigma)$ for $a, b\neq 0$, 
		then $Q^+(p) \subseteq A(\sigma)$,
	\item[]{\rm (2)}
		if $p = (0,b) \in A(\sigma)$ for $b\neq 0$, 
		then $H^+(p) \subseteq A(\sigma)$, and
	\item[]{\rm (3)}
		if $p = (a,0) \in A(\sigma)$ for $a\neq 0$, 
		then $V^+(p) \subseteq A(\sigma)$.
	\end{itemize}
\end{lemma}
\begin{proof}
	Let $C$ be a chain with turn sequence $\sigma$ that reaches $p=(a,b)$.
	If $a\neq 0$ and $b \neq 0$, we can reach any point in $Q^+(p)$, i.e.,
	$p+ (\sgn(a) \cdot i, \sgn(b) \cdot j)$ for $i,j \in \Zo$, 
	by stretching $C$ by $i$ units using the vertical cut $x = \sgn(a)/2$ and 
	by $j$ using the horizontal cut $y = \sgn(b)/2$.
	See Figure~\ref{fig:stretch}(a).
	
	If $p$ is on the $y$-axis, i.e., $p=(0,b)$ for $b\neq 0$,
	as in Figure~\ref{fig:stretch}(c), 
	we can reach any point in $H^+(p)$, i.e., 
	$p+ (i, \sgn(b) \cdot j)$ for $i \in \Z$ and $j \in \Zo$,
	by stretching $C$ by $j$ using
        the horizontal cut $y = \sgn(b)/2$,
	and by $|i|$ using
	one of two $y$-monotone cuts $f(y)$ or $-f(y)$
	depending on if the sign of $i$ is negative or positive.
	The cut $f(y)$ is $f(y)=1/2$ except for the domain 
	$y \in [-1/2,1/2]$ where it is $f(y) = 2|y| - 1/2$.
        Since the origin is adjacent to a horizontal segment (to the east),
	$f(y)$ and $-f(y)$ intersect only horizontal segments and separate $p$ and $o$.


	If $p$ is on the $x$-axis as in Figure~\ref{fig:stretch}(d),
	we can reach any point in $V^+(p)$, i.e., $p+ (\sgn(a) \cdot i, j)$
	for $i \in \Zo$ and $j \in \Z$,
	by stretching $C$ by $i$ using the vertical cut $x=\sgn(a)/2$,
	and by $|j|$ using
	one of two $x$-monotone cuts $g(x)$ or $-g(x)$
	depending on if the sign of $j$ is negative or positive.
	Let $t$ be the $x$-coordinate of the first bend in $C$, which
        exists since $|\sigma|>0$.
	The cut $g(x)$ is 
	$g(x) = -1/2 - 2x$ for $x \in [-1/2,0]$, 
	$g(x) = -1/2$ for $x \in [0,t]$, 
	$g(x) = -1/2 + 2(x-t)$ 	for $x \in [t,t+1/2]$,
	and $g(x) = 1/2$ otherwise.
	Since the origin has no adjacent horizontal segment to its west and
	the point $(t,0)$ has no adjacent horizontal segment to its east, 
	both $g(x)$ and $-g(x)$ intersect only vertical segments and
        separate $p$ and $o$.
	This concludes the proof.
\end{proof}

A \emph{staircase} is a turn sequence of $n$ alternating 
left and right turns.
One can easily observe
using the Stretching Lemma (Lemma~\ref{lem:stretch})
that $A(\sigma)$ for a staircase $\sigma$ is
identical to the quadrant $Q^+(p)$ where $p$ is the point reached by
the realization of $\sigma$ with unit length segments.
%
For a non-staircase sequence $\sigma$, we show that if $(a, b)\in
A(\sigma)$ then $(a, 0)\in A(\sigma)$ or $(0, b)\in A(\sigma)$.

\begin{lemma}[Axis Lemma]
\label{lem:unreachCorner}
For any turn sequence $\sigma$, except a staircase,
if $(a,b)$ is reachable by $\sigma$ then at least one of 
$(0, b)$ and $(a, 0)$ is reachable by $\sigma$.
\end{lemma}
\begin{proof}
	If $a=0$ or $b=0$ there is nothing to show.
	Let $C$ be a chain with turn sequence $\sigma$ that reaches $(a,b)$
	with non-zero $a$ and $b$.
	We show that at least one of $(0, b)$
	and $(a, 0)$ is reachable by $\sigma$.
	We do this by constructing two cuts of the chain $C$ 
	at least one of which will succeed in separating $(a,b)$ 
	from the origin $o$ and will allow us to stretch $C$ 
	to reach $(0, b)$ or $(a, 0)$.

	Let $\epsilon$ be a small positive real number.
	We define the $\epsilon$-extended upper (resp., lower) side 
	of a horizontal segment $(u,v)$ to be
        $(u+(-\epsilon,\epsilon),v+(\epsilon,\epsilon))$, i.e.,
	the translation of the ($2\epsilon$-lengthened) segment up
        (resp., down) by $\epsilon$.
	Similarly, we define the $\epsilon$-extended right (resp., left) side 
	of a vertical segment.

	We define two cuts $f(y)$ and $g(x)$ for $C$.
	To simplify the description, we will assume that $a >0$ and $b >0$.
	The other cases are similar. See Figure~\ref{fig:twoCuts}.
	$g(x)$ is the $x$-monotone cut such that 
	$g(x)=-\epsilon$ for $x \in (-\infty,0]$,
	and $g(x)$ is never closer than $\epsilon$ (in $x$ or $y$-coordinate) 
	to a horizontal segment of $C$.
	For $x > 0$, $g(x)$ is the maximum $y$-coordinate 
	subject to these constraints.
	It follows that we can view $g(x)$ as a staircase curve whose
        finite length horizontal segments are
	subsegments of the ($\epsilon$-extended lower sides of)
	horizontal segments of $C$.
	Each such horizontal segment of $g(x)$ ends at 
	$p_i +(\epsilon,-\epsilon)$ where $p_i$ is a bend point of $C$ 
	at the right end of a horizontal segment.

	Let $C'$ be the reflection of $C$ about the diagonal $x=y$, and
	$g'(x)$ be the $x$-monotone cut (defined above) of $C'$.
	$f(y)$ is the reflection of $g'(x)$ about the diagonal $x=y$.
	It follows that we can view $f(y)$ as a staircase curve whose finite
        length vertical segments are subsegments of
	the ($\epsilon$-extended left sides of) vertical segments of $C$.
	Each such vertical segment of $f(y)$ ends at $p_i +(-\epsilon,\epsilon)$
	where $p_i$ is a bend point of $C$ at the upper end of a vertical segment.
	In fact, both $g(x)$ and $f(y)$ are staircase curves monotone 
	to the $x$-axis and $y$-axis.

\begin{figure}
	\centering
	\includegraphics[scale=0.9]{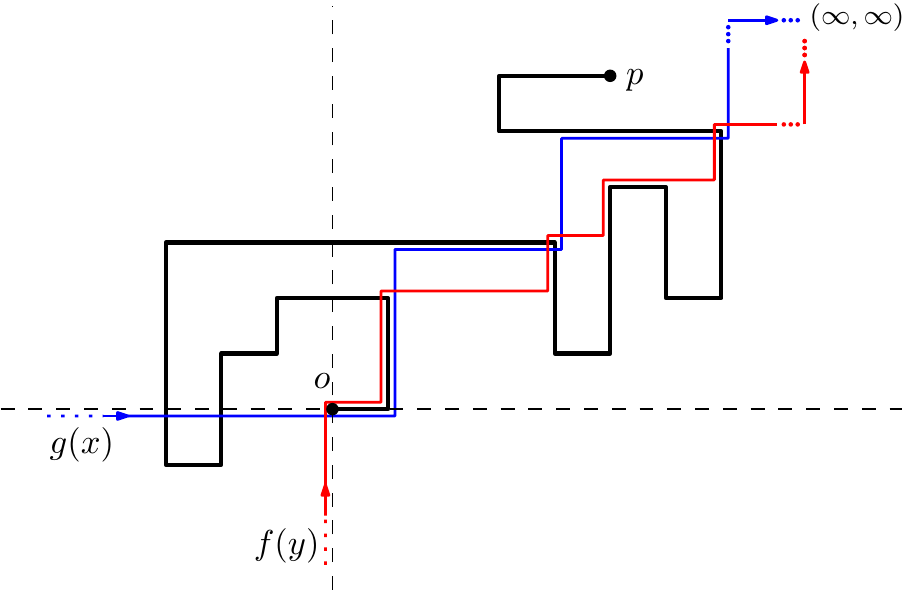}
	\caption{The two cuts $f(y)$ and $g(x)$ used in the proof of Lemma~\ref{lem:unreachCorner}.
	$f(y)$ separates $o$ and $p$, but $g(x)$ does not.}
	\label{fig:twoCuts}
\end{figure}

	If $g(x)$ separates $o$ and $p$, i.e., $p$ is below $g(x)$, 
	then we can stretch $C$ using the cut
	$g(x)$ so that $(a,0)$ is reached since $o$ is above $g(x)$.
	If $f(y)$ separates $o$ and $p$, i.e., $p$ is above $f(y)$
	then we can stretch $C$ using the cut
	$f(y)$ so that $(0,b)$ is reached.
	To prove the lemma, we must show that $p$ is below $g(x)$ or
	above $f(y)$.
	To do this we show the following claim:

\begin{claim}
\label{claim:orderCuts}

	Suppose that a chain $C$ realizing a non-staircase $\sigma$
	reaches $p=(a,b)$ where $a$ and $b$ are positive integers.
	If $p$ is above $g(x)$ then $p$ is above $f(y)$.
\end{claim}
\begin{claimproof}
	For $i=1, \dots, n$, let $p_i=(x_i,y_i)$ be the point 
	where $C$ makes the $i$th turn $\sigma_i$.
	Let $\sigma_j$ be the first turn where the staircase property is
	violated, so $\sigma_{j-1}\sigma_j = \lt\lt$ or 
	$\sigma_{j-1}\sigma_j = \rt\rt$.

	From $o$ to $p_j$, $f(y)$ and $g(x)$ follow opposite sides of
	every segment of $C$.
	Since $p=(a,b)$ is an integral point, if $a \le x_j$ then $p$
	above $g(x)$ implies $p$ above $f(y)$.

	If $p_j$ is a left turn, $f(y)$ turns right at 
	$p_j + (-\epsilon,\epsilon)$ while $g(x)$ continues upward.
	If $p_j$ is a right turn, $f(y)$ continues rightward 
	while $g(x)$ turns left at $p_j + (\epsilon,-\epsilon)$.
	In either case, until $f(y)$ and $g(x)$ intersect again, 
	if $p$ is above $g(x)$ then $p$ is above $f(y)$.

	Suppose now that $f(y)$ and $g(x)$ intersect again.
	Since both are staircase curves, and $g(x)$ is above $f(y)$, 
	this can occur only if a horizontal segment of $g(x)$ crosses 
	a vertical segment of $f(y)$.
	Since both horizontal segments of $g(x)$ and 
	vertical segments of $f(y)$ are subsegments of 
	($\epsilon$-extended sides of) segments of $C$, 
	this intersection occurs within $\epsilon$ of a bend point $p_k$ 
	which is the right endpoint of a horizontal
	segment of $C$ and the upper endpoint of a vertical segment of $C$.
	It follows that at $p_k + (\epsilon,-\epsilon)$, the curve $g(x)$
	turns upward while at $p_k + (-\epsilon,\epsilon)$, 
	the curve $f(y)$ turns rightward.
	Thus $g(x)$ and $f(y)$ intersect twice within $\epsilon$ of $p_k$
	and $g(x)$ continues above $f(y)$.
	Since $p$ is an integral point and not a bend point of $C$, 
	$p$ is not within $\epsilon$ of $p_k$ and the claim follows.
\end{claimproof}
	By this claim, if $p$ is above $g(x)$, then $p$ is above $f(y)$,
	so we can stretch $C$ using the cut $f(y)$ 
	so that $(0, b)$ is reachable. 
	Otherwise, since $p$ is integral, $p$ is below $g(x)$, and 
	we may stretch $C$ using the cut $g(x)$ so that $(a,0)$ is reachable.
\end{proof}

\section{Reachable set \texorpdfstring{$A(\sigma)$}{A(sigma)}}
\label{sec:reachable_set}
 
We now show, using the two lemmas in the previous section, that 
the reachable set $A(\sigma)$ is defined as 
the union of at most four halfplanes
whose bounding lines are orthogonal to the four signed axes.
For this, we define the \emph{Box algorithm}
and the \emph{TwoBox algorithm} for realizing 
$\sigma=\sigma_1\cdots\sigma_n$ of length $n$ 
as a simple rectilinear chain $C$.

\begin{figure}[t]
\centering
    \includegraphics[width=\textwidth]{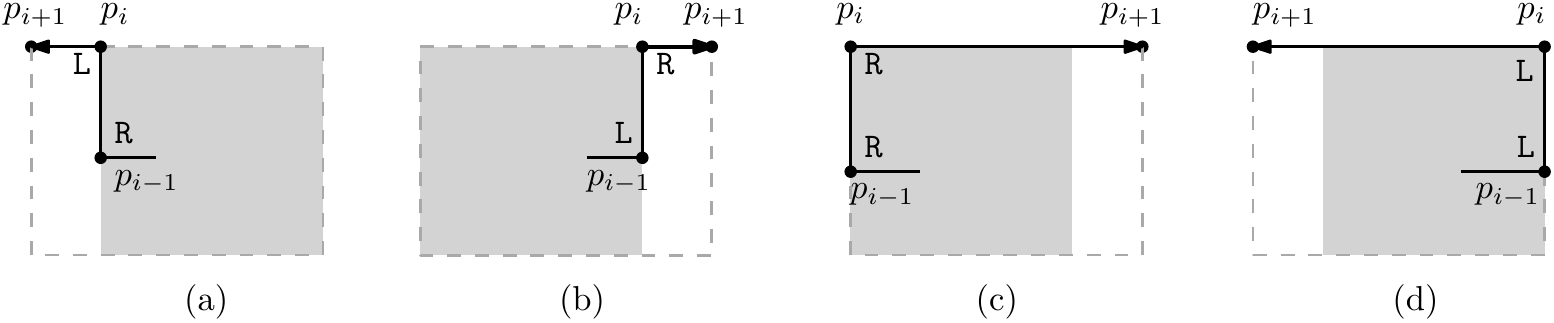}
\caption{ Box algorithm: (a) $\sigma_{i-1}\sigma_i = \rt\lt$. (b) $\sigma_{i-1}\sigma_i=\lt\rt$. (c) $\sigma_{i-1}\sigma_i = \rt\rt$. (d) $\sigma_{i-1}\sigma_i = \lt\lt$.}
\label{fig:box_algo}
\end{figure}
 
\subsection{Box algorithm}
The Box algorithm draws $\sigma$ as a chain $C$ in an incremental way, 
starting with a unit horizontal base segment $C_0$ from $o$ to $(1,0)$
and constructing a chain $C_{i}$ for the subsequence 
$\sigma_1\cdots\sigma_{i}$ from the chain $C_{i-1}$ 
for $\sigma_1\cdots\sigma_{i-1}$
by adding a segment that realizes the $i$th turn $\sigma_{i}$
for $1 < i \leq n$.
The chain $C_i$ has $i$ bend points, $p_1, \ldots, p_{i}$;
a starting point $p_0 = o$, and an endpoint $p(C_i) = p_{i+1}$. 
Each bend point $p_j$ for $1\leq j \leq i$ corresponds to 
the turn $\sigma_j$ on $C_i$. 

Each $C_i$ satisfies three invariants: 
(I1) the smallest bounding box $B_i$ of $C_i$ has dimension 
$\lceil \frac{i+1}{2}\rceil \times \lceil \frac{i+1}{2} \rceil$ 
for odd $i$ and
$\lceil \frac{i+1}{2}\rceil \times \lfloor \frac{i+1}{2} \rfloor$ 
for even $i$, 
(I2) the endpoint $p_{i+1}$ of $C_i$ is at a corner of $B_i$, 
and (I3) at least one side of the box that is incident to $p_{i+1}$ 
is not occupied by any other segments of $C_i$. 
Note that $C_0$ fits in the (degenerated) box 
with width of one and height of zero, 
so the three invariants are clearly satisfied.

To get $C_i$ from $C_{i-1}$, we determine the position of $p_{i+1}$, and
connect it to $p_i$, which is located at a corner of $B_{i-1}$, 
with the segment $e_i=p_ip_{i+1}$. See Figure~\ref{fig:box_algo}. 
Without loss of generality, we assume that $i$ is even, 
so $e_i$ is horizontal because $e_0$ is assumed to be horizontal. 
We further assume that 
$p_i$ lies at one of the two corners on the upper side of
the bounding box $B_{i-1}$ of $C_{i-1}$. 

We have two cases: whether $\sigma_{i}$ is different 
from $\sigma_{i-1}$ or not. 
If they are different, then $e_i$ can be a unit segment
as in Figure~\ref{fig:box_algo}(a)-(b).
%
%
When they are same, as in Figure~\ref{fig:box_algo}(c)-(d), 
we draw $e_i$ as a horizontal segment whose length is 
the width of $B_{i-1}$ plus one. 
The box $B_i$ is one unit wider than $B_{i-1}$. 
It is easy to check that $C_i$ and $B_i$ indeed satisfy 
the three invariants for both cases.

\subsection{TwoBox algorithm} 
The TwoBox algorithm splits $\sigma$ into two subsequences 
$\sigma' = \sigma_1\cdots \sigma_i$, and 
$\sigma'' = \sigma_{i+1}\cdots \sigma_n$ for some $1 \leq i < n$. 
We define $\bar{\sigma}''$ as the sequence of opposite turns of $\sigma''$ 
in the reverse order, i.e., 
$\bar{\sigma}'' = \bar{\sigma}_n\cdots \bar{\sigma}_{i+1}$, 
where $\bar{\sigma}_j$ is the opposite turn from $\sigma_j$. 

\begin{figure}[t]
	\centering
   \includegraphics{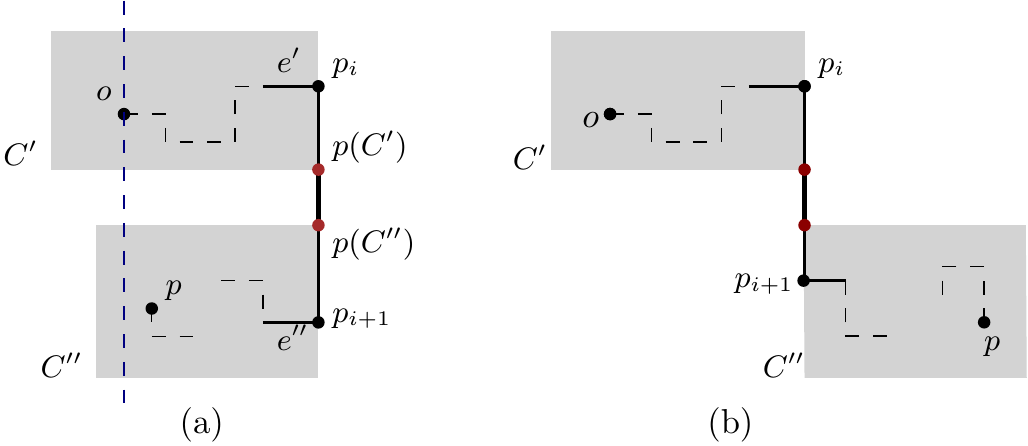}
	\caption{TwoBox algorithm: 
	(a) $\sigma_i\sigma_{i+1} = \rt\rt$. 
	(b) $\sigma_i\sigma_{i+1} = \rt\lt$.}
	\label{fig:two_box_algo}
\end{figure}

The algorithm draws a chain $C'$ for $\sigma'$ and 
another chain $C''$ for $\bar{\sigma}''$ by the Box algorithm. 
Their last segments should be connected to get the final chain $C$ 
by connecting the two endpoints $p(C')$ and $p(C'')$.
See Figure~\ref{fig:two_box_algo}. 
By the invariant (I2) of the Box algorithm, 
each endpoint lies at a corner of its box. 
We first place the boxes of $C'$ and $C''$ 
so that their last segments are aligned 
along a common (horizontal or vertical) line. 
Such alignment can be achieved by rotating $C''$ (if necessary). 
Note that rotating a chain does not affect its turn sequence. 
After the alignment, we simply connect 
the two endpoints $p(C')$ and $p(C'')$ by a unit segment,
then $p_i$ of $C'$ is finally connected with $p_{i+1}$ of $C''$ 
as the $(i+1)$th segment of $C$.
Note here that this connection is always possible 
due to the invariant (I2), and 
the two endpoints $p(C')$ and $p(C'')$ lie in the interior 
of the $(i+1)$th segment of $C$, thus both disappear.
 
A key property is that $C'$ and $C''$ are separable either 
by a horizontal or vertical cut. 
The bounding box of the resulting chain $C$ has dimensions each 
at most $\lceil (n+1)/2\rceil + 2$ 
by the first invariant of the Box algorithm.

\subsection{Hook patterns and axis reachability}

By the Stretching Lemma (Lemma~\ref{lem:stretch}), 
if a turn sequence $\sigma$ can reach a point $p$ on an axis, 
then a halfplane $H^+(p)$ or $V^+(p)$ is also in $A(\sigma)$.
Thus, to figure out the shape of $A(\sigma)$,
it is important to know the closest reachable point from $o$
on each signed axis.
Of course, a turn sequence may not reach 
any point on some signed axis.
We will show that axis reachability is determined by
hook patterns in the sequence.

A \emph{hook pattern} is either $\lt\lt$ or $\rt\rt$ in $\sigma$. 
This is realized in a chain $C$ by $\sigma$ as 
three consecutive segments 
whose two bend points correspond to $\lt\lt$ or $\rt\rt$. 
We call the middle segment the \emph{hook segment} in the chain $C$. 
We use the term \emph{hook} to indicate both a hook pattern and its
associated segment in a chain.

We assign a direction to the segments of $C$ along $C$ 
from $o$ to its endpoint of $C$.
According to its direction, we classify a hook into
four different types, 
\emph{up}, \emph{down}, \emph{left}, and \emph{right}
as shown in Figure~\ref{fig:relation_hooks}.
Up and down hooks are vertical segments, and 
left and right hooks are horizontal ones. 
Each hook has two subtypes according to its turns, $\lt\lt$ or $\rt\rt$. 
We illustrate all eight hook types 
in Figure~\ref{fig:relation_hooks}.

\begin{figure}[ht]
\centering
    \includegraphics[width=0.8\textwidth]{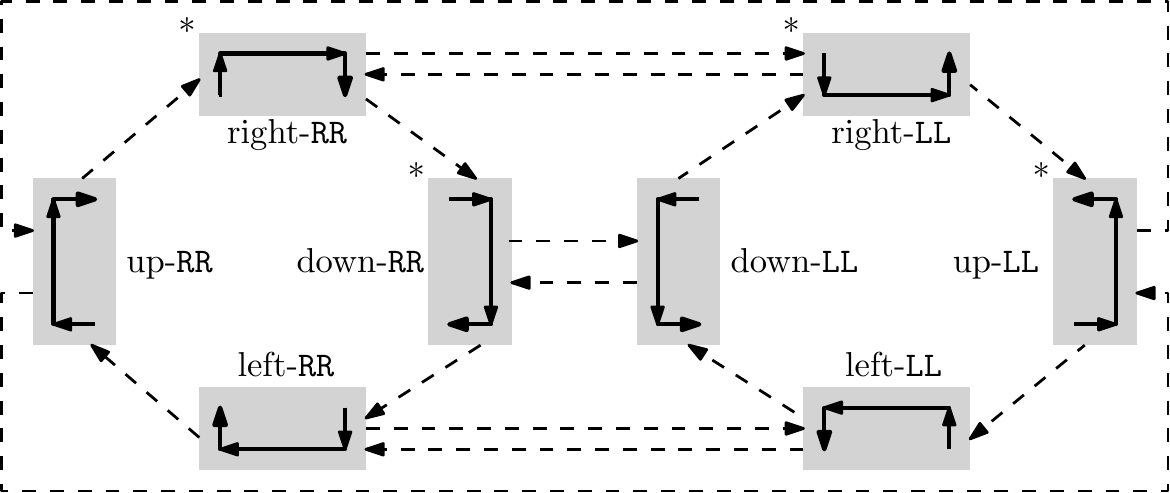}
	\caption{Four hook types: up, down, left, and right. 
	Each type has an $\lt\lt$ and $\rt\rt$ subtype. 
	A dashed arrow from a hook $h$ to $h'$ means that
	$h'$ can be the next hook after $h$ in the chain (they can be connected
	by a staircase of length $\geq 0$).
        The four hooks marked $*$ can be the first hook after $o$.
	}
\label{fig:relation_hooks}
\end{figure}

The chain reverses its direction only at hooks.
For instance, if $C$ reaches a point $p$ on the $+y$-axis, 
then it must have left the $x$-axis (at a point other than $o$) and
headed to the $+y$-axis.
This allows us to observe that a turn sequence without up hooks
cannot reach any point on the $+y$-axis and
indeed reveals the relation 
between the hook type and axis reachability.
It is worth mentioning that the first hook in the chain
is one of the four hooks, 
right-$\rt\rt$, right-$\lt\lt$, down-$\rt\rt$, and up-$\lt\lt$ 
because the first segment of any chain is 
the horizontal segment to the east.
We also observe that the sequence of hooks in a turn sequence is
restricted;
for instance, the hook preceding a down-$\lt\lt$ hook must be
a left-$\lt\lt$ hook or a down-$\rt\rt$ hook.
Figure~\ref{fig:relation_hooks} describes these constraints.
%

Using the relation on hook patterns with TwoBox algorithm,
we can prove the necessary and sufficient conditions for
the axis reachability. 
\begin{lemma}
\label{lem:hook_pt_on_axis}
	A turn sequence $\sigma$ of $n>0$ turns contains at least one 
	up, down, right, or left hook pattern 
	if and only if 
	there is a reachable point on 
	$+y$-axis, $-y$-axis, $+x$-axis, or $-x$-axis, respectively.
\end{lemma}
\begin{proof}
	We first show that any turn sequence $\sigma$ 
	containing a down-$\rt\rt$ hook can reach the $-y$-axis.
	We draw a chain $C$ using the TwoBox algorithm by splitting
	$\sigma$ into $\sigma' = \sigma_1\cdots\sigma_i$ and 
	$\sigma''=\sigma_{i+1}\cdots\sigma_n$, where
	$\sigma_i\sigma_{i+1}=\rt\rt$ is the first down-$\rt\rt$ hook from
	$o$ (see Figure~\ref{fig:two_box_algo}(a)).
	Let $p$ be the endpoint of the resulting chain.
	To get $C$, we translate $C''$ horizontally either by stretching the
	segment $p_{i-1}p_i$ (if $p$ is to the left of $o$) or the segment
	$p_{i+1}p_{i+2}$ (if $p$ is to the right of $o$) so that
	$p$ has the same $x$-coordinate as $o$.
	We then stretch the segment $p(C')p(C'')$ until $p$ is on the $-y$-axis.
	For the case that $\sigma$ contains down-$\lt\lt$, 
	we can apply a similar procedure to move $p$ to the $-y$-axis. 
	We thus conclude that if $\sigma$ has a down hook, 
	then it can reach a point on the $-y$-axis.

	The same argument can be applied for the other hook types:
	If $\sigma$ contains up, right, or left hooks, then 
	there are reachable points on $+y$-axis, $+x$-axis, or $-x$-axis,
	respectively.

	We now claim that the reverse is also true.
	In the following, we show
	that if there is a point 
	on the $-y$-axis reachable by $\sigma$, 
	then $\sigma$ must contain at least one down hook, 
	either a down-$\lt\lt$ hook or a down-$\rt\rt$ hook. 
	The other cases can be proved similarly.

	For any chain $C$ realizing $\sigma$, we define a point $q$
	as the first intersection of $C$ with the $-y$-axis.
	Let $q'$ be the last intersection point of $C$ 
	with the $+y$-axis before $q$.
	It is possible that $q = p$ or $q' = o$.
	Without loss of generality, we assume that 
	the subchain $C'$ of $C$ from $q'$ to $q$ 
	is to the right of the $y$-axis. 
	By the definition, 
	$C'$ cannot intersect with the $y$-axis except at $q'$ and $q$.
	Let $e = ab$ be the rightmost vertical segment connecting 
	two bends $a$ and $b$ of $C'$, where $a$ is below $b$.
	Note that $C'$ lies completely in the vertical slab between
	the $y$-axis and the line extending $e$.
	We can easily see that $e$ is a down hook from $b$ to $a$, i.e.,
	$C'$ is a subchain connecting $q'$, $b$, $a$, and $q$ in this order.
	Otherwise, i.e., if it is an up hook from $a$ to $b$, then
	the subchain from $b$ to $q$ of $C'$ must intersect
	the subchain from $q'$ to $a$ of $C'$, 
	which is a contradiction. 
\end{proof}

Let $(x^+_\sigma, 0)$ and $(x^-_\sigma,0)$ be 
the closest reachable points by $\sigma$ on
$+x$-axis and $-x$-axis, respectively, if they exist.
Let $(0, y^+_\sigma)$ and $(0, y^-_\sigma)$ be 
the closest reachable points by $\sigma$ on
$+y$-axis and $-y$-axis, respectively, if they exist.
We are now ready to give a complete characterization 
on the shape of $A(\sigma)$ by using
the properties proved so far.
\begin{theorem}
\label{thm:reachAxis}
The reachable set $A(\sigma)$ of a turn sequence $\sigma$ of
length $n\geq 0$ is: 
\begin{equation*}
A(\sigma) =
\begin{cases}
\{ (a, 0) \mid a \geq 1\}
&
\text{if $n=0$,}\\
Q^+((1+\floor{n/2}, \ceil{n/2}))
&
\text{if $\sigma$ is a staircase, $\sigma_1 = \lt$,}\\
Q^+((1+\floor{n/2}, -\ceil{n/2}))
&
\text{if $\sigma$ is a staircase, $\sigma_1 = \rt$,}\\
\bigcup_{a \in \{ x^{+}_\sigma, x^{-}_\sigma \}} V^+((a,0))
\cup
\bigcup_{b \in \{ y^{+}_\sigma, y^{-}_\sigma \}} H^+((0,b))
&
\text{otherwise.}
\end{cases}
\end{equation*}
	The signed coordinates $x^{+}_\sigma$, $x^{-}_\sigma$, 
	$y^{+}_\sigma$, and $y^{-}_\sigma$ exist 
	if and only if 
	$\sigma$ has right, left, up, and down hooks, respectively.
\end{theorem}
\begin{proof}
	It clearly holds for $n = 0$.
	A staircase sequence cannot reach any signed axis
	because it has no hooks, so $A(\sigma)$ is 
	the quadrant $Q^+(p)$ (by the Stretching Lemma)
where $p$ is the point reached by
the realization of $\sigma$ with unit length segments.
	
	We now suppose that $\sigma$ has one or more hooks.
	Let $A'(\sigma)$ be the union of (at most four) halfplanes
	whose bounding lines pass (orthogonally) 
	through the closest reachable points
	on the axes (determined by the hook patterns in $\sigma$). 
	By Lemma~\ref{lem:stretch} and Lemma~\ref{lem:hook_pt_on_axis}, 
	such halfplanes are reachable,
	so their union $A'(\sigma)$ is also reachable, i.e., 
	$A'(\sigma) \subseteq A(\sigma)$.
	To prove that $A'(\sigma)\supseteq A(\sigma)$, 
	we suppose that there is a point $p = (a, b)$ 
	where $a\neq 0$ and $b\neq 0$
	such that $p \in A(\sigma)$, but $p\not\in A'(\sigma)$.
	By the Axis Lemma (Lemma~\ref{lem:unreachCorner}), either $(a, 0)$ or 
	$(0, b)$ is reachable.
	If $(a, 0)$ is reachable, then it should fall into the open interval
	between $x^-_\sigma$ and $x^+_\sigma$ on the $x$-axis, 
	otherwise $(a, 0) \in V^+((x^+_\sigma, 0)) \cup V^+((x^-_\sigma, 0))$,
	thus $(a,0)\in A'(\sigma)$, 
	which implies $(a,b)\in A'(\sigma)$ by the Stretching Lemma. 
	The point $(a, 0)$ in the open interval means that 
	$(a, 0)$ is closer to $o$ than $(x^-_\sigma, 0)$
	or $(x^+_\sigma, 0)$, which is a contradiction.
	The same argument holds for the case that 
	$(0, b)$ is reachable.
\end{proof}

%
%

\begin{theorem}
\label{thm:connected}
	For a turn sequence $\sigma$, the reachable set $A(\sigma)$ and 
	the unreachable set $\Z^2\setminus A(\sigma)$ are both connected. 
\end{theorem}
\begin{proof}
	For an empty turn sequence (i.e., of length zero) or 
	a turn sequence with no hooks, 
	$A(\sigma)$ is clearly connected. 
	Suppose that $\sigma$ has hooks.
	By Theorem~\ref{thm:reachAxis}, $A(\sigma)$ is the union
	of at most four halfplanes whose bounding lines are orthogonal
	to the signed axes. 
	This implies that $A(\sigma)$ is connected
	except in the case that $A(\sigma)$ is 
	the union of two parallel halfplanes.
	Suppose without loss of generality that 
	their bounding lines are parallel to the $x$-axis. 
	Then, by Theorem~\ref{thm:reachAxis}, $\sigma$ contains
	exactly two types of
	hooks: down and up.
	However, in order for the up and down hooks to be connected by the chain,
	$\sigma$ needs to have at least one left or right hook; 
	see Figure~\ref{fig:relation_hooks} 
	for an illustration of the ordering relation of the hooks in a chain.
	This is a contradiction.
	The connectedness of the unreachable set $\Z^2\setminus A(\sigma)$
	immediately follows from the shape of $A(\sigma)$.
\end{proof}

\remark{%
	We observe that any turn sequence with
	five or more left turns than right ones (or with
	five or more right turns than left ones)
	has all four types of hooks, so such turn sequences 
	can reach all four signed axes.
}

\section{Closest reachable points on signed axes}
\label{sec:closest}

The next obvious question is to find the four closest reachable points
$x^+_\sigma, x^-_\sigma, y^+_\sigma$ and $y^-_\sigma$
for a given turn sequence $\sigma$ if they exist.
We give upper and lower bounds on the distance
from $o$ to the closest reachable points on the signed axes
as a function of the number of left and right turns in $\sigma$,
the size of the maximal monotone prefix or suffix of $\sigma$, 
and the maximal staircase prefix or suffix of $\sigma$.

To emphasize the number of the turns in $\sigma$,
we use another notation $\sigma_{l, r}$ to denote 
a turn sequence with $l$ left turns and $r$ right turns, where $n = l+r$. 
We define the \emph{excess number} $\delta$ of $\sigma_{l,r}$, 
denoted by $\delta(\sigma_{l,r})$, 
as the excess number of the left turns in the turn sequence, i.e., 
$\delta(\sigma_{l,r})=l-r$.
We define 
the \emph{prefix excess number} $\delta_i$,
as the excess number of the first $i$ turns in $\sigma_{l,r}$, i.e., 
$\delta_i=\delta(\sigma_1\cdots\sigma_i)$, where $\delta_0=0$ and $\delta_{l+r}=l-r$.
We call it \emph{prefix number} in short.

We assign prefix numbers $\delta_0, \ldots, \delta_{l+r}$ to the segments of $C$ 
in the order from $o$ to $p$. 
Assume that a segment of $C$ is directed from $o$ to $p$. 
We call a segment of $C$ with prefix number $t$ a \emph{\fbox{$t$}-segment},
and call a segment that is directed in the $z$-direction a \emph{$z$-segment}
for $z\in \{ \pm x, \pm y\}$. 
We can easily observe that a \fbox{$t$}-segment is a $+x$-segment, $+y$-segment, 
$-x$-segment, or $-y$-segment if $t \equiv 0, 1, 2,$ or $3 \pmod{4}$, respectively.
For example, if $\sigma_{l,r}=\lt\lt\rt\rt\rt$, we have
$\delta_0=0, \delta_1=1, \delta_2=2, \delta_3=1, \delta_4=0,$ and $\delta_5=-1$. 
The fourth segment of $C$ has $\delta_3=1$, so it is a \fbox{$1$}-segment
and also a $+y$-segment.

Let us define $x^+_{l,r} = \max_{\sigma_{l,r}} |x^+_{\sigma_{l,r}}|$
over all possible turn sequences $\sigma_{l,r}$ with $l$ left turns and $r$ right turns.
Similarly, we define $x^-_{l,r}$, $y^+_{l,r}$, and $y^-_{l,r}$
as the maximum distance from $o$ to the closest reachable points
by any turn sequence $\sigma_{l,r}$ on the corresponding signed axes.

A subchain of a chain $C$ that realizes $\sigma_{l,r}$ 
is said to be $z$-monotone 
for $z\in \{ \pm x, \pm y\}$
if it has no $-z$-segments.
Let $m^o_{z}$ and $m^p_{z}$ denote 
the number of $z$-segments in the maximal $z$-monotone subchains 
respectively containing $o$ and $p$, where $z\in \{ \pm x, \pm y\}$.
Let ($z_1$, $z_2$)-\emph{staircase} (or ($z_2$, $z_1$)-\emph{staircase}) denote 
a staircase that is $z_1$-monotone and $z_2$-monotone, 
where $z_1, z_2\in \{ \pm x, \pm y\}$.
Note that $z_1$- and $z_2$-directions are not the same, nor the opposite.
For better understanding, we use the cardinal directions, NE-staircase 
for the north-east staircase, i.e., ($+y$, $+x$)-staircase, and 
NW-staircase, SE-staircase, and SW-staircase 
the ($+y$, $-x$)-staircase, ($-y$, $+x$)-staircase, and ($-y$, $-x$)-staircase, respectively.
Let $_{z_1}^{}m_{z_2}^{o}$ and $_{z_1}^{}m_{z_2}^{p}$ 
denote the number of $z_2$-segments in the maximal ($z_1$, $z_2$)-staircases of $C$
containing $o$ and $p$, respectively.

Finally, we denote by $C[u,v]$ a subchain from $u$ to $v$ of $C$, 
where $u$ and $v$ are the points of $C$ and $u$ precedes $v$, that is, $u$ is
closer to $o$ than $v$.

\subsection{Upper and lower bounds on \texorpdfstring{$y^-_{l,r}$}{y-}
and \texorpdfstring{$y^+_{l,r}$}{y+}}

We first explain the (upper and lower) bounds on $y^-_{l,r}$ and $y^+_{l,r}$
can be easily derived from the bounds on $x^-_{l,r}$ and $x^+_{l,r}$.

Let $C$ be a rectilinear chain that realizes a turn sequence
$\sigma_{l,r}=\sigma_1\cdots\sigma_n$ which reaches to 
a point $p = (0, -b)$ on the $-y$-axis for some $b > 0$.
Define a new chain $C'$ by rotating $C$ by $90$ degrees in counterclockwise direction.
Turns are invariant to the rotation, thus the turn sequence of $C'$ remains unchanged.
Then $p = (0, -b)$ on the $-y$-axis is mapped to a point $p'=(b,0)$ on the $+x$-axis, and
the first (horizontal) segment of $C$ containing $(1,0)$ is mapped to 
a vertical segment containing $(0,1)$. 
We now augment $C'$ with a horizontal segment from $(-1,0)$ to $(0,0)$,
which creates a new turn $\lt$ at $(0,0)$; this augmentation is always possible
without causing self-intersections and changing the position of $p'$ 
by stretching along a vertical cut
between $(-1,0)$ and $(0,0)$ if $(-1,0)$ is occupied by a part of $C'$.
This resulting chain $C'$ realizes a new turn sequence 
$\sigma'_{l+1, r}=\lt\sigma_1\sigma_2\cdots\sigma_n$
such that it reaches to the point $p" = (b+1, 0)$ on the $+x$-axis
when it starts from $(0, 0)$.
This shows that if $\sigma'_{l+1,r}$ reaches to a point $p"=(a, 0)$ on the $+x$-axis,
then $\sigma_{l,r}$ reaches to a point $p=(0, -a+1)$ on the $-y$-axis.
This implies that the bounds on the $-y$-axis are directly derived 
from the ones on the $+x$-axis.
Similarly, by rotating in clockwise direction, 
the bounds on the $+y$-axis are also derived from the ones on the $+x$-axis.
Therefore, from now on, we consider the bounds only on the $-x$- and $+x$-axis.

\subsection{Upper bounds on \texorpdfstring{$x^-_{l,r}$}{x-} and \texorpdfstring{$x^+_{l,r}$}{x+}}

We assume that $l-r \geq 0$.

The upper bounds on $x^-_{l,r}$ and $x^+_{l,r}$ can be obtained by giving an algorithm that
draws a chain whose endpoint is on the target axis.
We first introduce a simple algorithm, called $\lt\rt$-algorithm,
which will be used as a subroutine in the drawing algorithms we propose.
Similar algorithms were previously known 
for a given sequence of exterior angles, 
called the \emph{rl}-algorithm~\cite{cr-socg85},
and for a given label-sequence~\cite{Sack84PhD}.

\paragraph*{\texorpdfstring{$\lt\rt$}{LR}-algorithm}

The $\lt\rt$-algorithm runs in two phases. 
In the \emph{reduction phase}, we find a pattern 
$\lt\rt$ or $\rt\lt$ in $\sigma_{l,r}$, and delete it, 
then we get a shorter sequence $\sigma_{l-1,r-1}$. 
We repeat this until there is no such pattern
(see~\cite{VijayanWigderson85} for a similar procedure).
Then the final sequence becomes $\sigma_{l-r,0}$ 
by the assumption that $l\geq r$. 
Note that $\sigma_{l-r, 0} = \lt^{l-r}$ and
is drawn as a spiral-like chain $C_{l-r,0}$
that wraps around $o$ in the counterclockwise direction. 
In the \emph{reconstruction phase}, 
we draw the chain incrementally from $C_{l-r, 0}$ 
for $\sigma_{l-r, 0}$, by inserting back $\lt\rt$ and $\rt\lt$ patterns in the reverse order of the deletions in the reduction phase. 
We reconstruct the chain $C_{l-i+1, r-i+1}$ 
from the chain $C_{l-i, r-i}$ for $i = r, \ldots, 1$
by drawing the two additional segments (for inserting back an $\lt\rt$ or $\rt\lt$) 
along a newly inserted empty row and column.

We now look closely at the chain $C_{l, r}$ drawn 
by the $\lt\rt$-algorithm when the excess number is small, i.e., $l-r=0,1,2$. 
For $l-r=0$, the final sequence in the reduction phase 
would be $\sigma_{0,0}$, thus the base chain is just 
a unit horizontal segment which connects 
two endpoints $a$ and $b$. 
Since rows and columns in the reconstruction phase 
are inserted between $a$ and $b$, 
those two endpoints still remain on the opposite sides 
of the bounding box until the end of the algorithm. 
A similar property holds for the other cases where $l-r = 1, 2$. 
The difference is that the sides where $a$ and $b$ lie are 
orthogonal for $l-r = 1$, and the same for $l-r = 2$.
Thus we can represent $C_{l,r}$ for $l-r = 0, 1, 2$ 
as a black box having one entry point $a$ and 
one exit point $b$ on its sides. 
These boxes will be used later as building blocks 
to draw chains with large excess numbers.

Let us consider $C_{l,r}$ in a different view when $l-r = 0$.
If the last turn is $\rt$, 
then $C_{l, r}$ is a concatenation of $C_{l, r-1}$ and 
the last segment
which are connected by $\rt$.
Since $l-(r-1) = 1$, we can draw the last segment (of $C_{l,r}$)
along the top side (parallel to the first segment) of the bounding box of $C_{l, r-1}$. 
If the last turn is $\lt$,
then we can draw the last segment, in a symmetric way,
along the bottom side of the bounding box of $C_{l-1, r}$.
We can draw any turn sequence with $l-r\leq 0$ in a similar way, so 
we can summarize these properties as follows.

\begin{lemma}
\label{lem:l-r-small}
A turn sequence $\sigma_{l, r}$ with $l-r \in \{ 0, \pm1, \pm2\}$
has a realization by a chain $C_{l,r}$ with both endpoints on the
sides of its bounding box.
Moreover, if $l-r = 0$, then there is a realization
in which the last segment is
contained in one side of the bounding box.
\end{lemma}

\paragraph*{Notation and observation}

Consider a turn sequence $\sigma_{l,r}$. 
We define $i(w)$ and $j(w)$ as the indices of
the last turn and the first turn whose excess number is exactly $w$, respectively.
Then $0 \leq i(w), j(w) \leq l+r$ and $\delta_{i(w)} = \delta_{j(w)} = w$.
We observe that for any $i(w) \leq l+r-2$, 
if $\delta_{i(w)} < l-r$, then the two consecutive turns $\sigma_{i(w)+1}$ and
$\sigma_{i(w)+2}$ must be $\lt$, otherwise both must be $\rt$
because otherwise, there would be $\delta_a = w$ for some $a> i(w)$, 
which is a contradiction that $i(w)$ is the index of the last segment
with the excess number $w$.
By the same reasoning, for any $j(w)\geq 2$, 
if $\delta_{j(w)} < 0$, the two previous turns 
$\sigma_{j(w)-1}$ and $\sigma_{j(w)-2}$ must be $\rt$, otherwise both must be $\lt$.
We have the following observations.

\begin{observation}
\label{obs:ik}
For all $i(w) \leq l+r-2$, 
if $\delta_{i(w)} < l-r$, then $\sigma_{i(w)+1}\sigma_{i(w)+2} = \mathtt{L}\mathtt{L}$.
Otherwise, if $\delta_{i(w)} > l-r$, then $\sigma_{i(w)+1}\sigma_{i(w)+2} = \mathtt{R}\mathtt{R}$.
\end{observation}

\begin{observation}
\label{obs:ik2}
For all $j(w) \geq 2$,
if $\delta_{j(w)} < 0$, then $\sigma_{j(w)-2}\sigma_{j(w)-1} = \mathtt{R}\mathtt{R}$.
Otherwise, if $\delta_{j(w)} > 0$, then $\sigma_{j(w)-2}\sigma_{j(w)-1} = \mathtt{L}\mathtt{L}$.
\end{observation}

\subsubsection{Upper bounds on \texorpdfstring{$x^-_{l,r}$}{x-}}
\label{subsec:upper_bound_x_minus}

Let $C$ denote a chain that realizes a turn sequence $\sigma_{l,r}$.
To reach the $-x$-axis, we know by Lemma~\ref{lem:hook_pt_on_axis} that $\sigma_{l, r}$ must have at least one left hook, i.e.,
there exists a \fbox{$3$}-segment or \fbox{$-3$}-segment in $C$.
We will explain drawing algorithms that determine the endpoint $p$ 
on the $-x$-axis, which gives upper bounds on $x^-_{l,r}$, 
for the cases of $l-r = 0, 1, 2,$ and $l-r \geq 3$.

\begin{figure}[t]
\centering
\includegraphics[width=\textwidth]{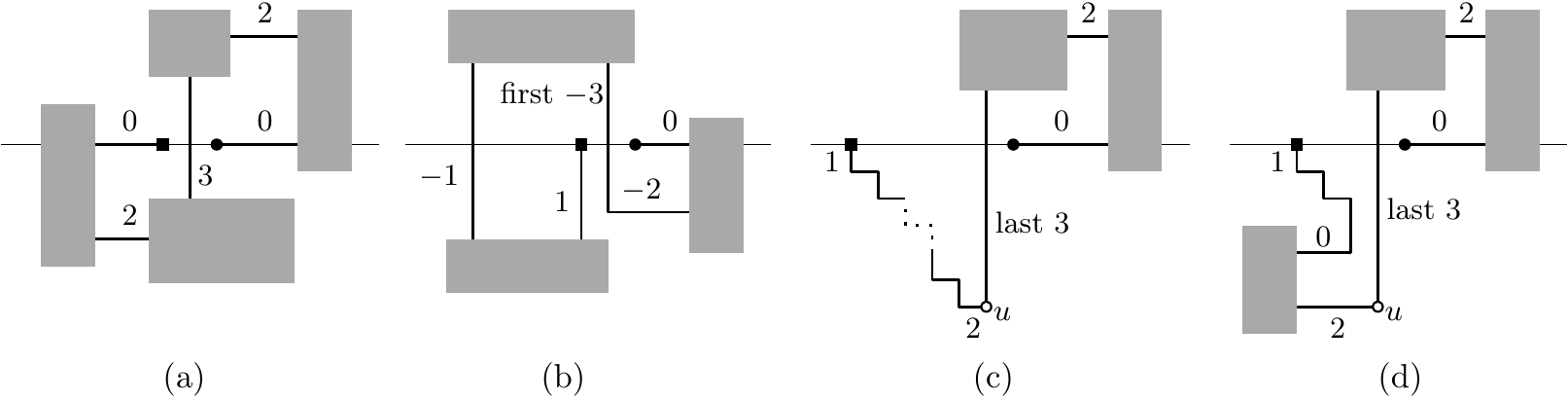}
\caption{Drawing algorithms of $C$ when $l-r = 0$ and $l-r = 1$:  
(a) Case 1.
(b) Case 2.1.
(c), (d) Case 2.2.
}
\label{fig:hyun1}
\end{figure}

\paragraph*{Case 1: $l-r=0$}

If there exists a \fbox{$3$}-segment in $C$, then we can draw $C$ as shown in Figure~\ref{fig:hyun1}(a). 
The black circle is $o$ and black square is $p$. 
The gray rectangles are bounding boxes enclosing subchains of $C$
which are drawn by $\lt\rt$-algorithm. 
The segments connecting the bounding boxes can be drawn without
any self-intersection because entry and exit points of the boxes
are on their sides by Lemma~\ref{lem:l-r-small} and 
are properly positioned
by stretching the boxes by horizontal or vertical cuts.
Locating \fbox{$3$}-segment at $x=-1$ allows $p$ to be at $(-2, 0)$, so
we have that $x^-_{l,r} \leq 2$. 
If there exist a \fbox{$-3$}-segment in $C$, we can bound $x^-_{l,r} \leq 2$ similarly by reflecting $C$ with respect to the $x$-axis. Thus we have that $x^-_{l,r}\leq 2$.

\paragraph*{Case 2: $l-r=1$}

We can draw $C$ in a similar way as Case 1. 

\subparagraph*{Case 2.1: There exist a \fbox{$-3$}-segment in $C$}

Draw $C$ as shown in Figure~\ref{fig:hyun1}(b) so that the first \fbox{$-3$}-segment crosses the $x$-axis at $(-1,0)$. 
The part after this \fbox{$-3$}-segment 
can be drawing like Figure~\ref{fig:hyun1}(b) so that $p$ is at $(-2,0)$. 
This is possible because the segment preceding the first \fbox{$-3$}-segment is a \fbox{$-2$}-segment. 
We thus have that $x^-_{l,r}\leq 2$.

\subparagraph*{Case 2.2: There exist no \fbox{$-3$}-segment in $C$}

For this case, there must be a \fbox{$3$}-segment. 
Select the last \fbox{$3$}-segment of $C$.
Let $u$ be its lower vertex. 
The drawing algorithm differs depending on whether $C[u,p]$ contains \fbox{$0$}-segment. 

If $C[u,p]$ contains no \fbox{$0$}-segments, then $C[u,p]$ only contains \fbox{$1$}-segments and \fbox{$2$}-segments, which is indeed a NW-staircase, so the position of $p$
is the width of the NW-staircase, i.e., the number of its horizontal segments, 
$_{-x}^{}m_{+y}^{p}$. See Figure~\ref{fig:hyun1}(c).
 
Otherwise, if $C[u,p]$ contains a \fbox{$0$}-segment, then select the last \fbox{$0$}-segment of the subchain. Its next two segments are 
\fbox{$1$}-segment and \fbox{$2$}-segment by Observation~\ref{obs:ik}, 
which form a NW-staircase up to $p$ as shown in Figure~\ref{fig:hyun1}(d).
 
For both cases, the width of the maximal NW-staircase containing $p$
affects the position of $p$. Since all segments of the staircase can be
drawn unit segments, we can locate $p$ at $(-(_{-x}^{}m_{+y}^{p}+1), 0)$,
so $x^-_{l,r}\leq \, _{-x}^{}m_{+y}^{p}+1$.

\paragraph*{Case 3: $l-r=2$}

We obtain a new chain $C'$ by deleting the last segment from $C$, 
then $\delta(C') = 1$ if the last turn is $\lt$ or $3$ otherwise.
After drawing $C'$ by algorithms used for the cases $l-r=1$ or $l-r=3$ which
will be explained below, adding a unit-length $-x$-segment at
the end of $C'$ gives a chain $C$. 
Thus the upper bound for $l-r=2$ becomes the upper bounds for $l-r=1$ or $l-r=3$ plus one.
We conclude that if $\sigma_n=\lt$, then $x^-_{l,r}\leq\, _{-x}^{}m_{+y}^{p}+2$, 
otherwise, if $\sigma_n=\rt$, then $x^-_{l,r}\leq 2$.

\paragraph*{Case 4: $l-r \geq 3$}

Consider the pairs
$\sigma_{i(w)+1}\sigma_{i(w)+2}=\lt\lt$
for $w=2,6,10,\ldots$.
To get a chain $C$, we split the turn sequence into
subsequences at those pairs,
draw them by the $\lt\rt$-algorithm, and merge them carefully.

At the first step, we take the subsequence 
$\sigma_1\cdots \sigma_{i(2)}$, and draw it as a chain $C_1$
by the $\lt\rt$-algorithm. Since $\delta_{i(2)} = 2$,
$C_1$ has entry and exit points on the left side of
its bounding box $B_1$ by Lemma~\ref{lem:l-r-small}.
We can place $B_1$ as in Figure~\ref{fig:l-r-algorithm-steps}(a) 
so that the entry point is on the $+x$-axis.
We connect $o$ with the entry point by a horizontal segment.
Since $\sigma_{i(2)+1}\sigma_{i(2)+2}=\lt\lt$ by Observation~\ref{obs:ik}, 
we can draw the corresponding segments so that 
the vertical segment passes through $(-1,0)$ 
as in Figure~\ref{fig:l-r-algorithm-steps}(a).
Note here that $\delta_{i(2)+1} = 3$ and $\delta_{i(2)+2}=4$.

\begin{figure}[th]
\centering
    \includegraphics[width=\textwidth]{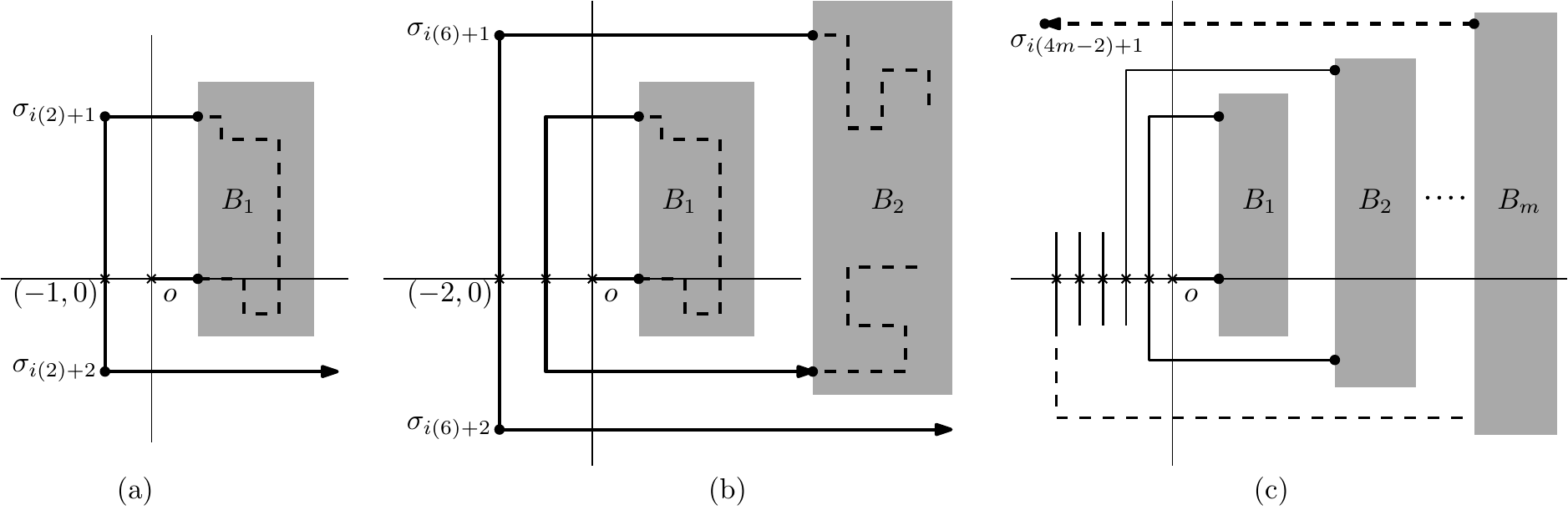}
\caption{The drawing steps for the turn sequence whose excess number is at least three.}
\label{fig:l-r-algorithm-steps}
\end{figure}

The next steps are clear. 
We take the next subsequence $\sigma_{i(2)+3}\cdots\sigma_{i(6)}$, 
draw its chain $C_2$, place the bounding box $B_2$
in the right of $B_1$, and draw the next three segments
corresponding to $\sigma_{i(6)+1}\sigma_{i(6)+2}$,
as in Figure~\ref{fig:l-r-algorithm-steps}(b),
without any self-intersection.
We repeat this until we place the bounding box $B_m$
of the chain $C_m$, 
where $m = \lfloor \frac{l-r+2}{4}\rfloor$. 
See Figure~\ref{fig:l-r-algorithm-steps}(c).
Note that $m\geq 1$ from the assumption $l-r\geq 3$.

The last step is to place $B_{m+1}$ of the chain $C_{m+1}$
for the remaining subsequence 
$\sigma^* = \sigma_{i(4m-2)+1}\cdots\sigma_{l+r}$ 
and draw the last segments
to reach the point $p = (-x, 0)$ with $x = \lfloor \frac{l-r+c}{4}\rfloor$
for some positive integer $c$.

We need to handle this step in different ways according to
the excess number $\delta(\sigma^*)\in \{0,1,2,3\}$.
It is easy to place $B_{m+1}$ and draw the last segment
when $\delta(\sigma^*) = 1$ or $\delta(\sigma^*)=2$ (equivalently, $l-r \equiv 3$ or $0 \pmod{4}$); 
see Figure~\ref{fig:l-r-algorithm-last}(a)-(b).
The final chain reaches $p = (-m,0)$.
When $\delta(\sigma^*) = 3$ (equivalently, $l-r \equiv 1 \pmod{4}$), we split it once again into two
subsequences $\sigma'$ and $\sigma''$ 
with $\sigma^*=\sigma'\sigma''$ 
such that $\delta(\sigma') = 2$ and $\delta(\sigma'') = 1$.
Draw $C'$ for $\sigma'$ and $C''$ for $\sigma''$ by
the $\lt\rt$-algorithm, and
place their boxes as in Figure~\ref{fig:l-r-algorithm-last}(c)
so that $p = (-m, 0)$.
For these three cases, we have that 
$x^-_{l,r} \leq m = \lfloor \frac{l-r+2}{4}\rfloor$.

\begin{figure}[th]
\centering
    \includegraphics[width=\textwidth]{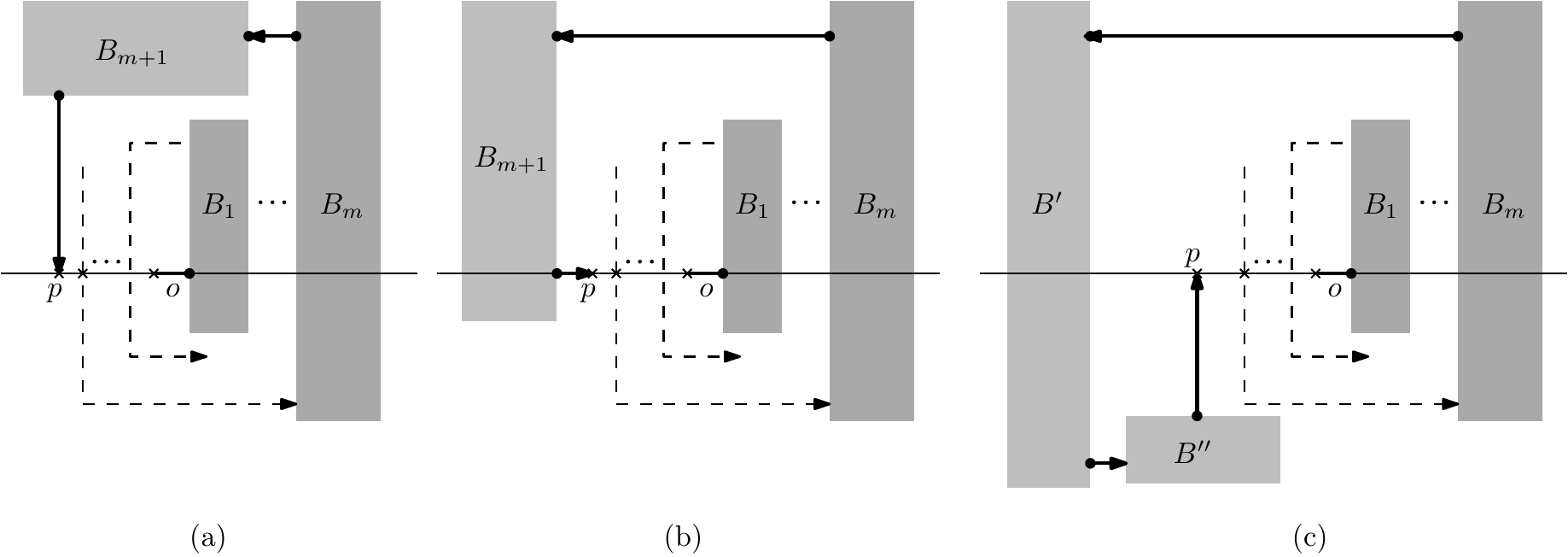}
\caption{Handling the last subsequence $\sigma^*$. 
	(a) When $\delta(\sigma^*)=1$, the last segment 
	enters $p=(-m, 0)$ from the north. 
	(b) When $\delta(\sigma^*)=2$, the last segment 
	enters $p=(-m, 0)$ from the west along the negative $x$-axis. 	
	(c) When $\delta(\sigma^*)=3$, $B_{m+1} = B' \cup B''$, 
	so the last segment enters $p=(-m, 0)$ from the south.}
\label{fig:l-r-algorithm-last}
\end{figure}

The last case that $\delta(\sigma^*) = 0$ (equivalently, $l-r \equiv 2 \pmod{4}$) 
should be handled more carefully.
Let $\sigma'$ be the subsequence obtained by deleting the last
turn $\sigma_n = \sigma_{l+r}$ from the original sequence $\sigma$.
Suppose first that $\sigma_{l+r} = \rt$.
Then $\delta(\sigma') = l-r+1$.
We draw $\sigma'$ in the way explained above,
then we can get a chain $C'$ as in Figure~\ref{fig:l-r-algorithm-last}(a) 
because $m = \lfloor \frac{l-r+3}{4}\rfloor$ and 
the excess number of the $(m+1)$th subsequence is now one.
The endpoint of $C'$ reaches $(-m, 0)$. 
To get $C$,
we simply extend $C'$ with a unit segment to the west
from its endpoint, which makes a bend for $\rt$. 
Then $C$ can reach $p=(-m-1,0)$,
thus $x^-_{l,r}\leq \lfloor \frac{l-r+3}{4}\rfloor+1 = \lfloor \frac{l-r+2}{4}\rfloor+1$.
Note that $l-r \equiv 2 \pmod{4}$.
For the other case that the last turn is $\lt$,
we can get $C'$ as in Figure~\ref{fig:l-r-algorithm-last}(c)
since $m = \lfloor \frac{l-r+1}{4}\rfloor$ and the excess
number of the $(m+1)$th subsequence is three.
We also add a unit segment from the endpoint
to the west to reach $p = (-m-1,0)$, 
thus $x^-_{l,r}\leq \lfloor \frac{l-r+1}{4}\rfloor+1 = \lfloor \frac{l-r+2}{4}\rfloor$.

\begin{theorem}
\label{thm:bounding-x-minus}
For any turn sequence $\sigma_{l,r}$ with $l-r \geq 0$, $x^-_{l,r}$ is bounded 
as follows:
\begingroup
\renewcommand{\arraystretch}{1.1} 
\begin{center}
\begin{tabular}{|ll|l|l|ll|ll|}
\hline
\multicolumn{2}{|l|}{}                           & $l-r=0$                  & $l-r=1$                                     & \multicolumn{2}{l|}{$l-r=2$}                                 & \multicolumn{2}{l|}{$l-r\geq 3$}                                                                                                                   \\ \hline
\multicolumn{2}{|l|}{\multirow{2}{*}{$x^-_{l,r}\leq$}} & \multirow{2}{*}{$2$} & \multirow{2}{*}{$_{-x}^{}m_{+y}^{p}+1$} & \multicolumn{1}{l|}{If $\sigma_n=\rt$} & $2$                    & \multicolumn{1}{l|}{\begin{tabular}[c]{@{}l@{}}If $\sigma_n=\rt$ and\\ $l-r\equiv 2 \pmod{4}$\end{tabular}} & $\left\lfloor \frac{l-r+2}{4}\right\rfloor+1$ \\ \cline{5-8} 
\multicolumn{2}{|l|}{}                                &                      &                                         & \multicolumn{1}{l|}{If $ \sigma_n=\lt$} & $_{-x}^{}m_{+y}^{p}+2$ & \multicolumn{1}{l|}{Otherwise}                                                                         & $\left\lfloor \frac{l-r+2}{4}\right\rfloor$   \\ \hline
\end{tabular}
\end{center}
\endgroup
\end{theorem}

\subsubsection{Upper bounds on \texorpdfstring{$x^+_{l,r}$}{x+}}
\label{subsec:upper_bound_x_plus}

We draw a chain in a similar way as the algorithm 
we did for bounding $x^-_{l,r}$.
Since the first segment of $C$ is a $+x$-segment, 
any \emph{$+x$-monotone} prefix (possibly suffix) of the turn sequence, 
i.e., a turn sequence with no vertical hook, 
can prevent the chain from
reaching a point close to $o$ on the $+x$-axis.
This is the main difference with the case of $x^-_{l,r}$.
 
First, consider a special situation that 
the whole sequence $\sigma_{l,r}$ is $+x$-monotone, i.e., 
no vertical hook in $\sigma_{l,r}$.
Note that this might happen only when $l-r = 0$ or $l-r = 1$. 
For this case, $m^o_{+x} = m^p_{+x}$. 
We can draw a chain with unit horizontal segments
so that $p$ is at $(m^o_{+x}, 0)$. 
We thus have that $x^+_{l,r} \leq m^o_{+x}$.

Otherwise, i.e., $\sigma_{l,r}$ has least least one vertical hook,
let $\sigma' = \sigma_1\cdots\sigma_g$
be the maximal $+x$-monotone prefix of $\sigma$
so that $\sigma_g\sigma_{g+1}$ is the first vertical hook.
As shown in Figure~\ref{fig:l-r-algo-for-xplus},
we have two cases: $\sigma_g\sigma_{g+1} = \lt\lt$
(up-$\lt\lt$ hook) or $\sigma_g\sigma_{g+1}=\rt\rt$ 
(down-$\rt\rt$ hook).
We draw $C'$ for $\sigma'$ with unit segments, and
$C''$ for the remaining subsequence $\sigma''$
by the winding scheme we used for bounding $x^-_{l,r}$.
Note that $l-r-2\leq \delta(\sigma'')\leq l-r+2$.

\begin{figure}[h]
\centering
    \includegraphics[width=\textwidth]{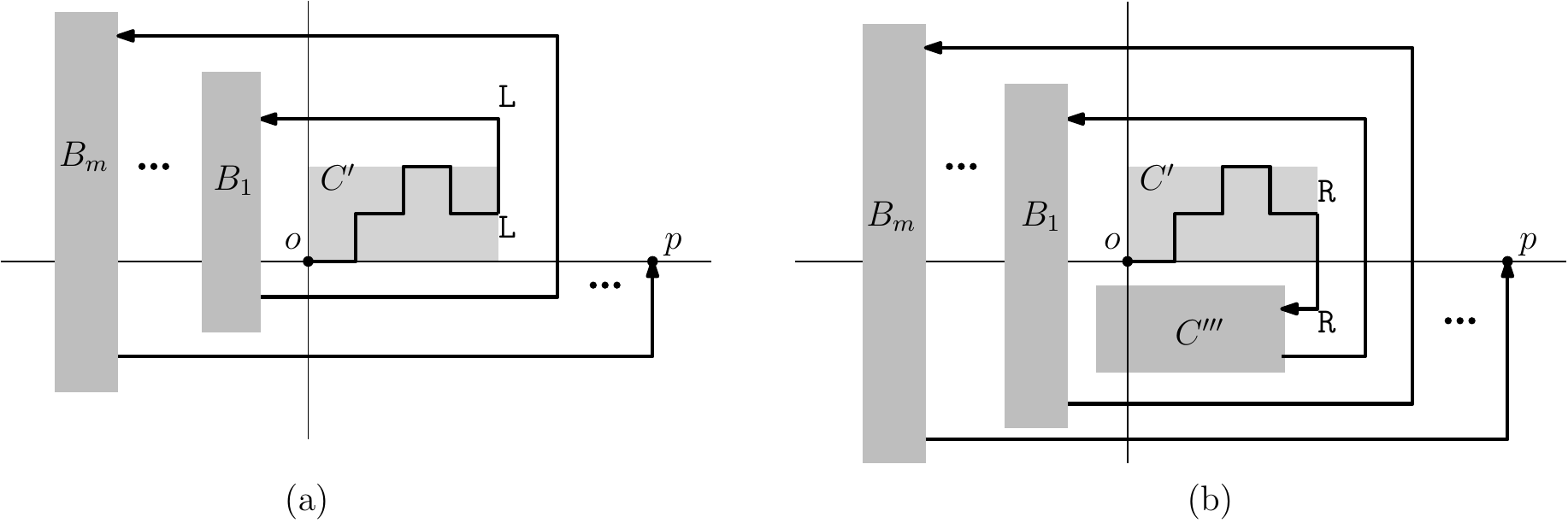}
\caption{Reaching the $+x$-axis. 
	Maximal $+x$-monotone prefixes followed by (a)  $\lt\lt$
	and (b) $\rt\rt$.}
\label{fig:l-r-algo-for-xplus}
\end{figure} 

For $\sigma_g\sigma_{g+1}=\rt\rt$, before applying 
the winding scheme, we draw the first part of $\sigma''$, that is,
$\sigma_{g+2}\cdots\sigma_{i(2)}$ as a chain $C'''$ using the
$\lt\rt$ algorithm,
then connect it with $C'$ by $\rt\rt$, and
extend its exit point to the bend corresponding to $\sigma_{i(2)+1}$.
The remaining steps are the same as before.

The length of $\overline{op}$ is determined by 
the width of $C'$ plus the number of vertical segments in $C''$ intersected 
by the $+x$-axis; the latter could be increased by one 
if the last segment of $C$ is a $+x$-segment, i.e., $l-r \equiv 0 \pmod 4$.
The former is just the number of $+x$-segments of $C'$,
denoted by $m^o_{+x}$, which is $\lceil |\sigma'|/2\rceil$.
The latter differs depending on the value of $l-r-2$.
If $l-r-2\geq 3$, then it is at most 
$\lfloor \frac{\delta(\sigma'')+2}{4}\rfloor+1 \leq 
\lfloor \frac{(l-r+2)+2}{4}\rfloor + 1=\lfloor \frac{l-r}{4}\rfloor +2$
by Theorem~\ref{thm:bounding-x-minus}.
Then we have that $x^+_{l,r} \leq m^o_x + \lfloor \frac{l-r}{4}\rfloor +2
\leq m^o_x + \lfloor \frac{l-r+2}{4}\rfloor +2$.
We will explain the remaining case that $-2 \leq l-r-2 < 3$, 
i.e., $0\leq l-r < 5$, later.

There is a way to improve this bound as follows.
We simply switch the role of $o$ and $p$, 
and draw a rectilinear chain $\bar{C}$ that realizes the turn sequence
$\bar{\sigma} = \bar{\sigma_n}\cdots\bar{\sigma_1}$
and starting from $p$, 
where $\bar{\sigma_i}$ is the opposite turn from $\sigma_i$.
Imagine that $p$ is the origin, then $o$ is now on the $-x$-axis.
If the first segment from $p$ is not a $+x$-segment 
(equivalently, $l-r \not\equiv 2 \pmod 4$),
divide $\bar{\sigma}$ into two subsequences
$\bar{\sigma'}$ and $\bar{\sigma''}$,
where $\bar{\sigma'}$ is the maximal $-x$-monotone prefix of $\bar{\sigma}$,
and $\bar{\sigma''}$ is the remaining subsequence.
We draw $\bar{C'}$ for $\bar{\sigma'}$ with unit segments,
and draw $\bar{C''}$ for $\bar{\sigma''}$, 
by the winding scheme we used for bounding $x^-_{l,r}$.
Note that $l-r-3 \leq \delta(\bar{\sigma''}) \leq l-r+3$.

The length of $\overline{op}$ is equal to
the width of $\bar{C'}$ 
plus the number of vertical segments in $\bar{C''}$ intersected by the $-x$-axis plus one, 
because the last segment of $\bar{C}$ is always a $-x$-segment.
The former is the number of $-x$-segments of $\bar{C'}$, 
which is the maximal $-x$-monotone subchain of $\bar{C}$ containing $p$.
We can easily observe that it is equal to the number of $+x$-segments in the maximal $+x$-monotone subchain of $C$ containing $p$, denoted by $m^p_{+x}$.
The latter differs depending on the value of $l-r-3$.
We will explain the case $-3 \leq l-r-3 < 3$, i.e., $0\leq l-r < 6$, later.
If $l-r-3\geq 3$, by Theorem~\ref{thm:bounding-x-minus}, we can bound $x^+_{l,r}$ as follows:

\begin{itemize}
\item If $\bar{\sigma_1} = \rt$, i.e., $\sigma_1 = \lt$,
then $x^+_{l,r} \leq 
m^p_{+x} + \lfloor \frac{\delta(\bar{\sigma''})+2}{4}\rfloor \leq 
m^p_{+x} + \lfloor \frac{l-r+2}{4}\rfloor + 1$.

\item If $\bar{\sigma_1} = \lt$, i.e., $\sigma_1 = \rt$,
then $x^+_{l,r} \leq
m^p_{+x} + \lfloor \frac{\delta(\bar{\sigma''})+2}{4}\rfloor+1 \leq 
m^p_{+x} + \lfloor \frac{l-r+2}{4}\rfloor +2$.
\end{itemize}

Otherwise, if the first segment from $p$ is a $+x$-segment 
(equivalently, $l-r \equiv 2 \pmod 4$), then we can apply
the winding scheme used for bounding $x^-_{l,r}$.
By Theorem~\ref{thm:bounding-x-minus}, we can bound $x^+_{l,r}$ as follows:

\begin{itemize}
\item If $l-r = 2$ and $\bar{\sigma_1} = \rt$, i.e., $\sigma_1 = \lt$, 
then $x^+_{l,r} \leq _{+x}^{}m_{+y}^{o} + 2$.

\item If $l-r = 2$ and $\bar{\sigma_1} = \lt$, i.e., $\sigma_1 = \rt$,
then $x^+_{l,r} \leq 2$.

\item If $l-r \geq 6$ and $\bar{\sigma_1} = \rt$, i.e., $\sigma_1 = \lt$,
then $x^+_{l,r} \leq \lfloor \frac{l-r+2}{4}\rfloor$.

\item If $l-r \geq 6$ and $\bar{\sigma_1} = \lt$, i.e., $\sigma_1 = \rt$,
then $x^+_{l,r} \leq \lfloor \frac{l-r+2}{4}\rfloor+1$.
\end{itemize}

We now have the following result. 

\begin{lemma}
\label{lem:bounding-x-plus}

For any turn sequence $\sigma_{l,r}$ with $l-r=2$ or $l-r \geq 6$, 
$x^+_{l,r}$ are bounded as follows:
\begingroup
\renewcommand{\arraystretch}{1.3} 
\begin{center}
\begin{tabular}{|l|l|l|l|l|} 
\hline
\multicolumn{2}{|l|}{}   & $l-r=2$                  & \begin{tabular}[c]{@{}l@{}}$l-r \geq 6$~and \\$l-r \equiv 2 \pmod 4$\end{tabular} & \begin{tabular}[c]{@{}l@{}}$l-r \geq 6$~and\\$l-r \not\equiv 2 \pmod 4$\end{tabular}  \\ 
\hline
\multirow{2}{*}{$x^+_{l,r}\leq$} & If $\sigma_1=\lt$ & $_{+x}^{}m_{+y}^{o} + 2$ & $\left\lfloor \frac{l-r+2}{4}\right\rfloor$                                       & $\min\{m^o_{+x}+2, m^p_{+x}+1\}+ \left\lfloor \frac{l-r+2}{4}\right\rfloor$           \\ 
\cline{2-5}
                                & If $\sigma_1=\rt$ & $2$                      & $\left\lfloor \frac{l-r+2}{4}\right\rfloor + 1$                                   & $\min\{m^o_{+x}+2, m^p_{+x}+2\}+ \left\lfloor \frac{l-r+2}{4}\right\rfloor$           \\
\hline
\end{tabular}
\end{center}
\endgroup
\end{lemma}

We now have remaining cases $l-r = 0, 1, 3, 4, 5$ which are not explained yet.
We will explain each case in the order of $l-r = 1, 0, 3, 5, 4$.

\paragraph*{Case 1: $l-r=1$}
If $\sigma$ is $+x$-monotone, i.e., no vertical hook,
then it can reach $p = (m^o_{+x}, 0)$ as mentioned earlier. 
We here suppose that $\sigma$ has a vertical hook.

We divide $\sigma$ into $\sigma'$ and $\sigma''$ as same as above;
$\sigma' = \sigma_1\cdots\sigma_g$ is the maximal $+x$-monotone prefix of $\sigma$
such that $\sigma_g\sigma_{g+1}$ is the first vertical hook, 
and $\sigma''=\sigma_{g+1}\cdots\sigma_n$.
Note that $l-r-2 \leq \delta(\sigma'') \leq l-r+2$, 
i.e., $-1 \leq \delta(\sigma'')\leq 3$.
We now have two subcases according to the existence of the right hook in $\sigma''$.

\subparagraph*{Case 1.1: $\sigma''$ contains a right hook.}
If $\delta_{g+1}=-2$, i.e., $\sigma_g\sigma_{g+1}=\rt\rt$,  
then $\delta(\sigma'') = \delta_{l+r} - \delta_{g+1} = 3$.
As in Figure~\ref{fig:l-r-algo-for-xplus}(b),
draw a chain $C'''$ below the $x$-axis.
The last segment from $C'''$ is \fbox{$0$}-segment, so
we can place $p$ just right to $C'$ on the $x$-axis after a final left turn.
This guarantees that $x^+_{l,r} \leq m^o_{+x}+1$. 
Otherwise, if $\delta_{g+1}=2$, i.e., $\sigma_g\sigma_{g+1}=\lt\lt$, 
then $\delta(\sigma'') = \delta_{l+r} - \delta_{g+1} = -1$.
The method here is a symmetric and $180$-degrees rotating version of the one 
to reach the $-x$-axis when $l-r=1$ in Figure~\ref{fig:hyun1}(b)-(d).
The difference is the existence of \fbox{$5$}-segments.
If \fbox{$5$}-segment exists, then we draw four bounding boxes by $\lt\rt$-algorithm,
place them as in Figure~\ref{fig:tarkcase1.1}(a). 
This gives $x^+_{l,r} \leq m^o_{+x}+2$.
For the other case that no \fbox{$5$}-segment exists,
the distance to $p$ can be affected as well by the length of the staircase
as in Figure~\ref{fig:tarkcase1.1}(b).
Since \fbox{$-1$}-segment exists in $\sigma''$, we draw the last \fbox{$-1$}-segment
so that it passes the $+x$-axis in the right of $C'$.
Two situations can occur depending on the existence of \fbox{$2$}-segments; 
if not exist, it directly reaches to $p$ via a NE-staircase (drawn as a black chain),
otherwise it goes around then take a NE-staircase after 
the last \fbox{$2$}-segment (drawn as red dashed chain and box).
For this case, the distance to $p$ is determined by 
the width of the maximal $+x$-monotone chain plus the width of the NE-staircase
containing $p$.
Thus, we have that
$x^+_{l,r} \leq m^o_{+x}+\, _{+x}^{}m_{+y}^{p}+1$.

\subparagraph*{Case 1.2: $\sigma''$ contains no right hook.}

In this case, we cannot apply the winding scheme used before.
We first know that $\delta_{g+1} \neq -2$;
otherwise $\delta(\sigma'') = \delta_{l+r} - \delta_{g+1} = 1 - (-2) = 3$, which
means there exist a right hook in $\sigma''$, a contradiction.
We now have that $\delta_{g+1}=2$.
Moreover, for all $g+1 < k < l+r$, $\delta_k \ne -1$ holds.
Because if there exist such $k$ that $\delta_{k} = -1$, 
then $\delta_{k} - \delta_{g+1} = 3$, so a right hook exists.
We conclude that after the turn $\sigma_{i(2)+1}$, 
there exist only \fbox{$0$}-segments and \fbox{$1$}-segments in $C$, and 
these segments form a NE-staircase.
We can draw $C$ as shown in Figure~\ref{fig:tarknew2}(a), and we have
$x^+_{l,r} \leq _{+x}^{}m^o_{+y} + _{+x}^{}m_{+y}^{p}+1$.
In particular, if $\sigma_1 = \rt$, then we have $x^+_{l,r} \leq \, _{+x}^{}m_{+y}^{p}+1$
because $_{+x}^{}m^o_{+y} = 0$.

\vspace*{2mm}
Now, let us divide $\bar{\sigma}$ into 
$\bar{\sigma'} = \bar{\sigma_n} \cdots \bar{\sigma_h}$
and $\bar{\sigma''} = \bar{\sigma_{h-1}} \cdots \bar{\sigma_1}$,
which are the longest $+x$-monotone prefix of $\bar{\sigma}$,
equivalently, the longest $+x$-monotone suffix of $\sigma$, 
and the remaining subsequence, respectively. We have two subcases
according to the existence of the left hook of $\bar{\sigma''}$.

\begin{figure}[t]
\centering
\includegraphics[]{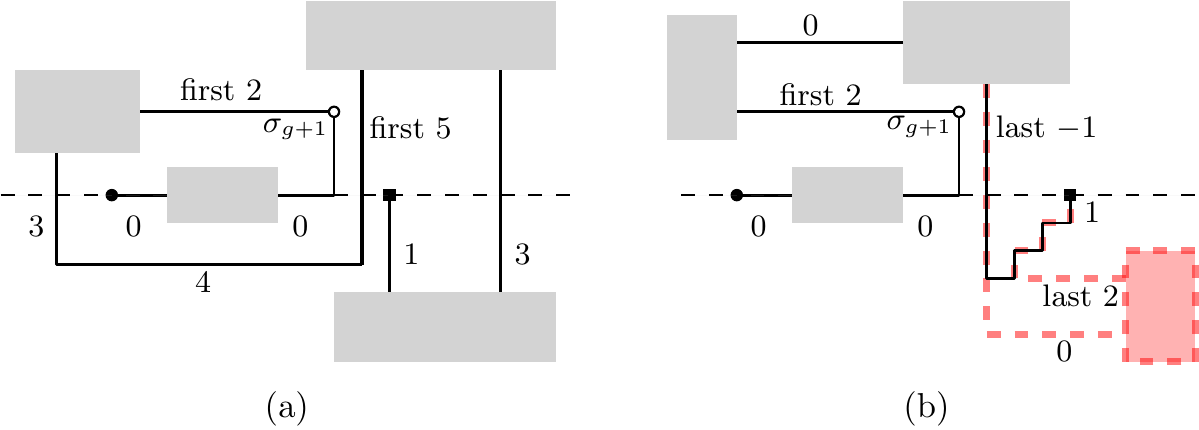}
\caption{Case 1.1 that $\sigma''$ has a right hook and $\sigma_g\sigma_{g+1}=\lt\lt$:  
(a) $C$ has a \fbox{$5$}-segment.
(b) $C$ has no \fbox{$5$}-segment.
 }
\label{fig:tarkcase1.1}
\end{figure}

\subparagraph*{Case 1.3: $\bar{\sigma''}$ contains a left hook}

Applying the same method we used in Case 1.1, 
We can have that if $\sigma_1 = \lt$, then
$x^+_{l,r} \leq m^p_{+x}+\,_{+x}^{}m_{+y}^{o}+2$,
otherwise, if $\sigma_1 = \rt$, then $x^+_{l,r} \leq m^p_{+x}+\,_{+x}^{}m_{-y}^{o}+2$.

\begin{figure}[t]
\centering
\includegraphics[width=\textwidth]{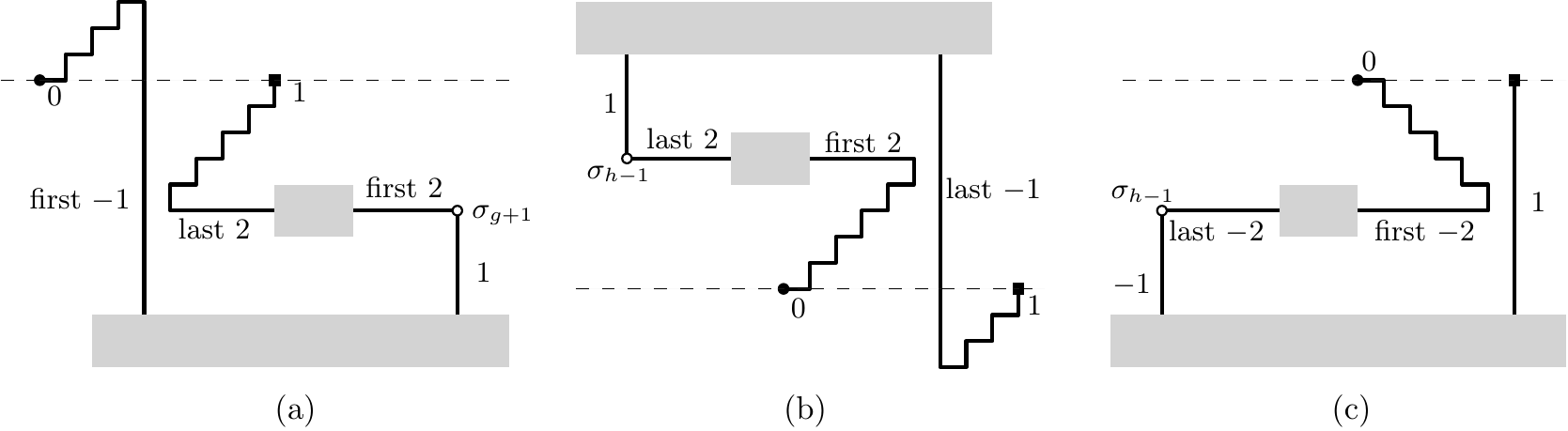}
\caption{Case 1 that $l-r=1$.
(a) Case 1.2.
(b) Case 1.4, $\sigma_1=\lt$.
(c) Case 1.4, $\sigma_1=\rt$.
 }
\label{fig:tarknew2}
\end{figure}

\subparagraph*{Case 1.4: $\bar{\sigma''}$ contains no left hook}

In this case, we need another drawing algorithm other than 
the winding scheme used for bounding $x^-_{l,r}$.
If $\sigma_1 = \lt$, then we have $\delta_{h-2} = 2$ and
$\delta_k \ne -1$ for all $1 < k < h-2$.
We know that before the turn $\sigma_{j(2)}$, 
there exist only \fbox{$0$}-segments and \fbox{$1$}-segments in $C$, and 
they form a NE-staircase.
We can draw $C$ as shown in Figure~\ref{fig:tarknew2}(b), and we have
$x^+_{l,r} \leq _{+x}^{}m^o_{+y} + _{+x}^{}m_{+y}^{p}+1$.

Otherwise, if $\sigma_1 = \rt$, we have $\delta_{h-2} = -2$ and 
$\delta_k \ne 1$ for all $1 < k < h-2$.
We conclude that before the turn $\sigma_{j(2)}$, 
there exist only \fbox{$0$}-segments and \fbox{$-1$}-segments in $C$, and
they form a SE-staircase.
We can draw $C$ as shown in Figure~\ref{fig:tarknew2}(c), and we get
$x^+_{l,r} \leq _{+x}^{}m^o_{-y} + 1$.

\vspace*{2mm}
We then have the following result for Case 1.

\begin{lemma}
\label{lem:bounding-x-plus2}
For any $+x$-monotone sequence $\sigma_{l,r}$ with $l-r=1$,
$x^+_{l,r} \leq m^o_{+x}$.
For any non-$+x$-monotone sequence $\sigma_{l,r}$ with $l-r = 1$,
if $\sigma_1 = \lt$, then
$$x^+_{l,r} \leq \min\{ m^o_{+x}+\, _{+x}^{}m_{+y}^{p}+1,\, m^p_{+x}+\, _{+x}^{}m_{+y}^{o}+2 \},$$
if $\sigma_1 = \rt$, then
$$x^+_{l,r} \leq \min\{ m^o_{+x}+\, _{+x}^{}m_{+y}^{p}+1,\, m^p_{+x}+\,_{+x}^{}m_{-y}^{o}+2 \}.$$
\end{lemma}

Consider a turn sequence $\sigma_{l,r}$ with $l-r = -1$ and a rectilinear chain $C$ that realizes $\sigma_{l,r}$ and whose endpoint $p$ is on the $+x$-axis.
Let $C^*$ be the reflection of $C$ with respect to the $x$-axis with endpoint $p$. 
The turn sequence $\sigma^*$ of $C^*$ is obtained by reversing $\sigma$,
so $\delta(\sigma^*)=1$ and $p=p^*$, 
so the length of $\overline{op}$ is equal to the length of $\overline{op^*}$.
Thus we have the following result.

\begin{lemma}
\label{lem:bounding-x-plus2-1}
For any $+x$-monotone sequence $\sigma_{l,r}$ with $l-r=-1$,
$x^+_{l,r} \leq m^o_{+x}$.
For any non-$+x$-monotone sequence $\sigma_{l,r}$ with $l-r = -1$,
if $\sigma_1 = \lt$, then
$$x^+_{l,r}\leq \min\{ m^o_{+x}+\,_{+x}^{}m_{-y}^{p}+1,\, m^p_{+x}+\,_{+x}^{}m_{+y}^{o}+2 \},$$
if $\sigma_1 = \rt$, then
$$x^+_{l,r}\leq \min\{ m^o_{+x}+\,_{+x}^{}m_{-y}^{p}+1,\, m^p_{+x}+\,_{+x}^{}m_{-y}^{o}+2 \}.$$
\end{lemma}

Lemma~\ref{lem:bounding-x-plus2} and Lemma~\ref{lem:bounding-x-plus2-1} 
will be used for handling Case 2.

\paragraph*{Case 2: $l-r=0$}
If we delete the last turn $\sigma_{l+r}$ from the original sequence $\sigma$, 
then its excess number becomes $1$ or $-1$.
Thus we can use Lemma~\ref{lem:bounding-x-plus2} and Lemma~\ref{lem:bounding-x-plus2-1}
to bound $x^+_{l,r}$.

\begin{lemma}
\label{lem:bounding-x-plus2-2}
For any $+x$-monotone sequence $\sigma_{l,r}$ with $l-r=0$,
$x^+_{l,r} \leq m^o_{+x}$.
For any non-$+x$-monotone sequence $\sigma_{l,r}$ with $l-r = 0$,
if $\sigma_1 = \lt$, then
$$x^+_{l,r} \leq \min\{ m^o_{+x}+\max\{_{+x}^{}m_{+y}^{p},\, _{+x}^{}m_{-y}^{p}\}, 
m^p_{+x}+ \, _{+x}^{}m_{+y}^{o}\}+2,$$
if $\sigma_1 = \rt$, then
$$x^+_{l,r} \leq \min\{ m^o_{+x}+\max\{_{+x}^{}m_{+y}^{p},\, _{+x}^{}m_{-y}^{p}\},
m^p_{+x}+ \, _{+x}^{}m_{-y}^{o}\}+2.$$
\end{lemma}

\begin{figure}[t]
\centering
\includegraphics[]{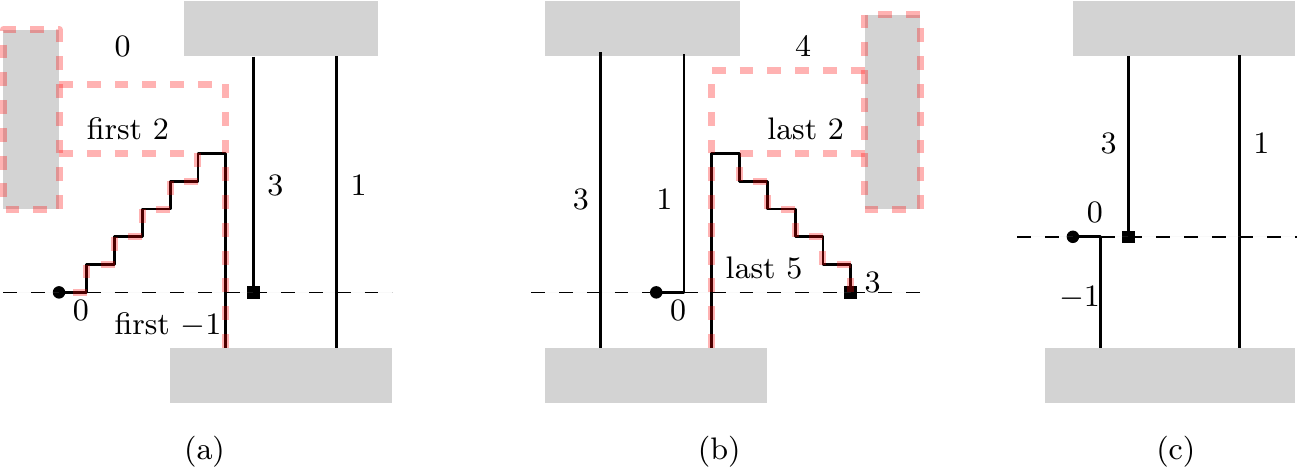}
\caption{Case 3 that $l-r=3$:
(a) Case 3.1, $C$ has a \fbox{$-1$}-segment.
(b) Case 3.1, $C$ has a \fbox{$5$}-segment.
(c) Case 3.2.
}
\label{fig:tark3}
\end{figure}

\paragraph*{Case 3: $l-r=3$}

We have two subcases depending on $\sigma_1$.

\subparagraph*{Case 3.1: $\sigma_1 = \lt$.}

By Lemma~\ref{lem:hook_pt_on_axis}, 
$C$ contains a \fbox{$-1$}-segment or a \fbox{$5$}-segment to reach the $+x$-axis.

If $C$ contains a \fbox{$-1$}-segment, 
we can draw $C$ as shown in Figure~\ref{fig:tark3}(a).
If $C$ has a \fbox{$2$}-segment before the first \fbox{$-1$}-segment,
then it is drawn along the dashed red route, otherwise along the black one, 
which gives $x^+_{l,r} \leq _{+x}^{}m_{+y}^{o}+2$.

If $C$ contains a \fbox{$5$}-segment,
we can draw $C$ as shown in Figure~\ref{fig:tark3}(b).
If $C$ has a \fbox{$2$}-segment after the last \fbox{$5$}-segment,
then it is drawn along the dashed red route, otherwise along the black one.
We have that $x^+_{l,r} \leq _{+x}^{}m_{-y}^{p}+2$.

If $C$ contains both \fbox{$-1$}-segment and \fbox{$5$}-segment,
then we simply take the drawing that gives a smaller distance to $p$,
so we have that
$x^+_{l,r} \leq \min\{_{+x}^{}m_{+y}^{o},\, _{+x}^{}m_{-y}^{p}\}+2$.

\subparagraph*{Case 3.2: $\sigma_1 = \rt$.}
We can draw $C$ as shown in Figure~\ref{fig:tark3}(c). We have that $x^+_{l,r} \leq 2$.

We then have the following result for Case 3.

\begin{lemma}
\label{lem:bounding-x-plus3}
For any turn sequence $\sigma_{l,r}$ with $l-r = 3$, $x^+_{l,r}$ is bounded as follows:
\begingroup
\renewcommand{\arraystretch}{1.1} 
\begin{center}
\begin{tabular}{|l|l|l|l|l|} 
\hline
\multirow{3}{*}{$x^+_{l,r}\leq$} & exists? & \fbox{$-1$}-segment & \fbox{$5$}-segment  & \fbox{$-1$}- and \fbox{$5$}-segments                                        \\ 
\cline{2-5}
                                & If~$\sigma_1=\texttt{L}$                                                                                & $_{+x}^{}m_{+y}^{o}+2$ & $_{+x}^{}m_{-y}^{p}+2$ & $\min\{_{+x}^{}m_{+y}^{o},\, _{+x}^{}m_{-y}^{p}\}+2$  \\ 
\cline{2-5}
                                & If~$\sigma_1=\texttt{R}$                                                                                & \multicolumn{3}{l|}{$2$}                                                                              \\
\hline
\end{tabular}
\end{center}
\endgroup
\end{lemma}

\begin{figure}[h]
\centering
\includegraphics[]{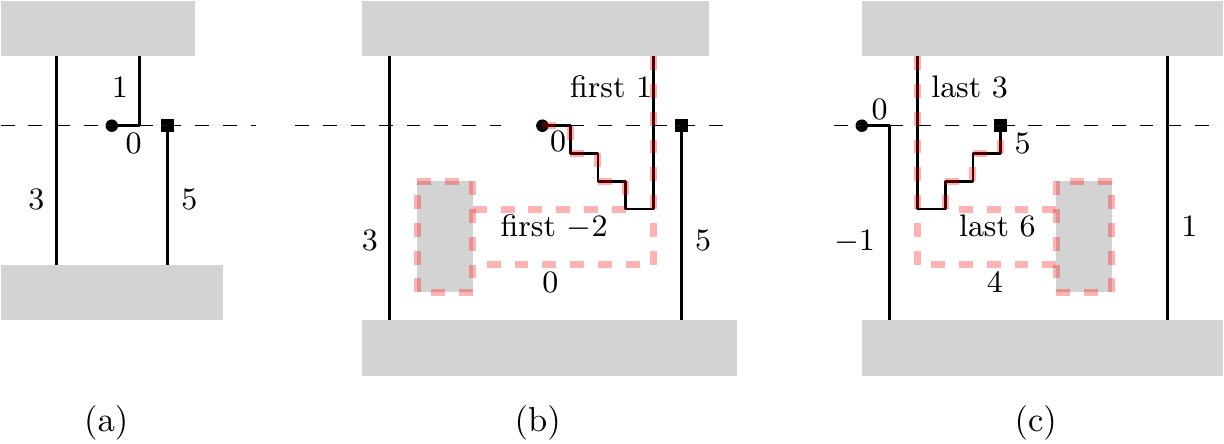}
\caption{Case 4 that $l-r=5$:
(a) Case 4.1.
(b) Case 4.2, the first \fbox{$1$}-segment of $C$ crosses $+x$-axis.
(c) Case 4.2, the last \fbox{$3$}-segment of $C$ crosses $+x$-axis.
}
\label{fig:tark4}
\end{figure}

\paragraph*{Case 4: $l-r=5$}

We have two subcases depending on $\sigma_1$.

\subparagraph*{Case 4.1: $\sigma_1 = \lt$}
We can draw $C$ as shown in Figure~\ref{fig:tark4}(a). 
Then we have that $x^+_{l,r} \leq 2$.

\subparagraph*{Case 4.2: $\sigma_1 = \rt$}
We can draw $C$ as shown in Figure~\ref{fig:tark4}(b)-(c).
In Figure~\ref{fig:tark4}(b), if $C$ has a \fbox{$-2$}-segment 
before the first \fbox{$1$}-segment, 
then it drawn along the dashed red route, otherwise along the black one.
We have that $x^+_{l,r} \leq _{+x}^{}m_{-y}^{o}+2$. 

Similarly, in Figure~\ref{fig:tark4}(c), if $C$ has a \fbox{$6$}-segment 
after the last \fbox{$3$}-segment,
then it drawn along the dashed red route, otherwise along the black one.
We have that $x^+_{l,r} \leq _{+x}^{}m_{+y}^{p}+2$.

We then have the following result for Case 4.

\begin{lemma}
\label{lem:bounding-x-plus4}
For any turn sequence $\sigma_{l,r}$ with $l-r = 5$, if $\sigma_1 = \lt$, 
then $x^+_{l,r} \leq 2$. 
Otherwise, $x^+_{l,r} \leq \min\{ _{+x}^{}m_{-y}^{o},\, _{+x}^{}m_{+y}^{p}\}+2$.
\end{lemma}

\paragraph*{Case 5: $l-r=4$}
If we delete the last turn $\sigma_{l+r}$ from the original sequence $\sigma$, 
then its excess number becomes $3$ or $5$,
so we can use Lemma~\ref{lem:bounding-x-plus3} and Lemma~\ref{lem:bounding-x-plus4}
to bound $x^+_{l,r}$.

\begin{lemma}
\label{lem:bounding-x-plus3,4-1}
For any turn sequence $\sigma_{l,r}$ with $l-r = 4$, $x^+_{l,r}$ is bounded as follows:
\begingroup
\renewcommand{\arraystretch}{1.1} 
\begin{center}
\begin{tabular}{|l|l|l|l|l|} 
\hline
\multirow{3}{*}{$x^+_{l,r}\leq$} & exists? & \fbox{$-1$}-segment  & \fbox{$5$}-segment  & \fbox{$-1$}- and \fbox{$5$}-segments                                        \\ 
\cline{2-5}
                                & If~$\sigma_1=\texttt{L}$                                       & $_{+x}^{}m_{+y}^{o}+3$ & $_{+x}^{}m_{-y}^{p}+3$ & $\min\{_{+x}^{}m_{+y}^{o},\, _{+x}^{}m_{-y}^{p}\}+3$  \\ 
\cline{2-5}
                                & If~$\sigma_1=\texttt{R}$                                       & \multicolumn{3}{l|}{$\min\{_{+x}^{}m_{-y}^{o},\, _{+x}^{}m_{+y}^{p}\}+3$}                               \\
\hline
\end{tabular}
\end{center}
\endgroup
\end{lemma}

We conclude this section with the following theorem.

\begin{theorem}
\label{thm:x-plus-conclusion}
For any turn sequence $\sigma_{l,r}$ with $l-r \geq 0$, 
$x^+_{l,r}$ can be bounded by
referring Lemma~\ref{lem:bounding-x-plus}, Lemma~\ref{lem:bounding-x-plus2}, Lemma~\ref{lem:bounding-x-plus2-2}, Lemma~\ref{lem:bounding-x-plus3}, Lemma~\ref{lem:bounding-x-plus4},
and Lemma~\ref{lem:bounding-x-plus3,4-1}.
\end{theorem}


\subsection{Lower bounds on \texorpdfstring{$x^+_{l,r}$}{x+} and \texorpdfstring{$x^-_{l,r}$}{x-}}
\label{subsec:lower_bound_x_plus}

We now show a lower bound on the distance of the closest reachable point to the origin on $x$-axis. 
Let $C$ be an arbitrary rectilinear chain that realizes a turn sequence $\sigma_{l,r}=\sigma_1 \sigma_2 \cdots \sigma_n$ for $l-r\geq 0$ which is reachable to the $x$-axis, 
and $p$ be the endpoint of $C$ that is on the $x$-axis.
First, we will bound the minimum number of vertical segments of $C$ that intersect $\overline{op}$, which gives a lower bound on the distance. 
For this, we use the \emph{rotation number} of a \emph{whisker-free} polygon,
which is introduced by Gr{\"u}nbaum~\cite{gs-tams90}. 
A polygon is whisker-free if no two edges incident 
with a vertex of the polygon overlap in a segment of positive length. 

We redefine the rotation number of a whisker-free polygon $P$ as follows.

\begin{definition}
Let $P$ be a whisker-free polygon that has $n$ vertices $v_1, v_2, \ldots, v_n=v_0$
and $n$ directed edges $e_1=(v_1, v_2), e_2=(v_2, v_3), \ldots, e_n=e_0=(v_n, v_1)$. 
Let $\alpha(v_i)$ denote the signed angle~\footnote{An angle measured from $e_{i-1}$ to $e_i$ in counterclockwise direction is positive, for clockwise direction negative.} 
between the direction vector of $e_{i-1}$ and $e_i$, and let $d(v_i) = \alpha(v_i)/2\pi$ denote the deflection of $v_i$ for $i = 1, 2, \ldots, n$.
The rotation number $R(P)$ of $P$ is the sum of deflections of vertices in $P$, i.e., 
$R(P) = \sum_{i=1}^n d(v_i)$.
\end{definition}

Gr{\"u}nbaum~\cite{gs-tams90} also introduced the \emph{ordinary} polygon;
a polygon is ordinary if no three edges have a common point. 
It is clear that every ordinary polygon is whisker-free.

A rectilinear polygon $P_C$ is defined by connecting $o$ and $p$ of $C$, 
but it is not necessarily ordinary because some segments of $C$ 
can be overlapped with $\overline{op}$. 
If the first or last segment of $C$ does, then define a subchain $C'$ of $C$ 
by deleting the overlapped segments.
Note that $P_{C'}$ obtained by connecting two endpoints of $C'$ is 
now whisker-free, but not ordinary yet
if there are segments of $C'$ overlapped with
the segment connecting the endpoints of $C'$.
We can make a new chain $C''$ from $C'$
by translating the overlapped segments by small positive amount above or below 
the $x$-axis while keeping the continuity of $C'$.
Then $P_{C''}$ obtained by connecting two endpoints of $C''$ becomes ordinary.

We will show that the minimum number of vertical segments of $C$ 
intersecting $\overline{op}$ is derived from the rotation number
of the ordinary polygon made in this way.
We now suppose that $p$ is on the $+x$-axis.

\subsubsection{The minimum number of vertical segments of $C$ that cross \texorpdfstring{$\overline{op}$}{op}}
\label{subsubsec:verticalnumber}

For any turn sequence $\sigma_{l,r}$ with $l-r \equiv 0 \pmod 4$, 
we have $\delta_{l+r-1}=l-r\pm1 \equiv \pm1 \pmod 4$,
The last segment of its chain $C$ is a $+x$-segment coming to $p$ from the west, that is, it is contained in the $+x$-axis.
Deleting the last segment gives a shorter chain $C^*$ whose 
$l-r \equiv 1 \pmod 4$ or $l-r \equiv 3 \pmod 4$.
By bounding the length of the last segment of $C$ into $1$,
we can make the lower bound on $\left| \overline{op} \right|$ of $C$
be bounded below by the lower bound on $\left| \overline{op^*} \right|$ of $C^*$ plus one, where $p^*$ is the endpoint of $C^*$.
We here consider only the turn sequence $\sigma_{l,r}$ with $l-r \not\equiv 0 \pmod 4$.

We have two cases $\sigma_1 = \lt$ and $\sigma_1=\rt$. 

\paragraph*{Case 1: $\sigma_1 = \lt$}
There are $l-r$ more left turns than right turns in $C$.
Let $k=\left\lfloor\frac{l-r}{4}\right\rfloor \geq 0$. 
Then $l-r \in \{4k+1, 4k+2, 4k+3\}$.
Let $P$ be an ordinary polygon derived from $C$ by connecting $p$ of $C$
and the second vertex $q$ of $C$ and removing the first segment $\overline{oq}$ of $C$.
Let us count the number of left and right turns of $P$.
The first turn $\lt$ at $q$ disappears in $P$.
The last segment $\overline{pq}$ of $P$ creates two turns at $p$ and $q$;
the turn $\rt$ at $q$ and the turn $\lt$ or $\rt$ at $p$ 
depending on $l-r\equiv 1 \pmod 4$ or $l-r\equiv 3 \pmod 4$, respectively.
Reflecting all these changes, we can conclude that $P$ has exactly 
$4k$ more left turns than right ones. 
Since a left turn and a right turn contributes to a deflection by $\frac{1}{4}$ 
and $-\frac{1}{4}$ respectively, 
we have $R(P) = 4k \times \frac{1}{4} = k = \lfloor\frac{l-r}{4}\rfloor$.

Whitney~\cite{w-cm37} defined the rotation number of an oriented smooth closed curve,
and Polyak~\cite{p-amst99} rewrote it simply as the number of rotations made by the tangent vectors as traversing along the curve.
Whitney~\cite{w-cm37} also introduced a \emph{normal} oriented smooth closed curve, and we redefine it as follows:
An oriented smooth closed curve is normal 
if the curve has neither overlapping nor touching parts, and 
no three or more pieces of the curve intersect at a common point
if it has self-intersections.

Let $S$ be an oriented smooth closed curve obtained by 
smoothening the vertices and their neighborhoods of $P$. 
Note that we can determine the smoothness
so that no self-intersection is generated. 
Since $P$ and $S$ are topologically equivalent, we have $R(P) = R(S)$, 
where $R(S)$ denote the rotation number of $S$.
Moreover $S$ is normal because $P$ is ordinary. 

Let $I$ be the total number of self-intersections of $S$ (also of $P$).
We show the relation between $R(S)$ and $I$ as follows.

\begin{lemma}
\label{lem:rot and self}
For a normal oriented smooth closed curve $S$, $I \geq \left| \left|R(S)\right|-1 \right|$.
\end{lemma}
\begin{proof}
A self-intersection of $S$ is the crossing between two different pieces of $S$.
Let $c_1$ and $c_2$ be the two pieces of $S$ that cross each other 
such that $c_1$ precedes $c_2$.
There are two types of self-intersections; it is the first type 
if $c_2$ crosses over $c_1$ from the left side to the right side of $c_1$, and
otherwise the second type.
Let $I^+$ and $I^-$ denote the number of the first and second type self-intersections, respectively. Then $I^+ + I^- = I$.
Whitney~\cite{w-cm37} proved $R(S) = \mu + I^+ - I^-$, where
$\mu$ is either $1$ or $-1$. 
By using this formula and the triangular inequality, we conclude that
$$I = I^+ +I^- = |I^+|+|I^-| \geq \left| I^+-I^-\right|
= \left|R(S)-\mu\right| \geq \left| \left|R(S)\right|-\left|\mu\right| \right|
= \left| \left|R(S)\right|-1 \right|.$$
\end{proof}

Because $R(P) = R(S)$ as we mentioned above, we have, by Lemma~\ref{lem:rot and self}, that
$$I\geq \left| R(P)-1 \right| = \left| \left\lfloor\frac{l-r}{4}\right\rfloor-1 \right|.$$

All self-intersections of $P$ are made on $\overline{pq}$, so 
at least $\left| \left\lfloor\frac{l-r}{4}\right\rfloor-1 \right|$ 
vertical segments of $P$ cross $\overline{pq}$. 
As the second segment of $C$ which is vertical also intersects (in fact, touches)
with $\overline{op}$, the number of vertical segments intersecting (or crossing) 
$\overline{op}$ is at least 
$$\left| \left\lfloor\frac{l-r}{4}\right\rfloor-1 \right|+1.$$

\paragraph*{Case 2: $\sigma_1 = \rt$}
By the same argument as in Case 1, we have $R(P) = (4k+4) \times \frac{1}{4} = k+1 = \lfloor\frac{l-r}{4}\rfloor + 1$.
By Lemma~\ref{lem:rot and self}, at least 
$\left\lfloor\frac{l-r}{4}\right\rfloor +1$ vertical segments of $C$ cross 
$\overline{op}$, including the second segment of $C$.

We conclude this section with the following theorem.

\begin{theorem}
\label{thm:vertical-seg-crossing}
For any turn sequence $\sigma_{l,r}$ with $l-r \not\equiv 0 \pmod 4$ reachable to the $+x$-axis,
any rectilinear chain $C$ that realizes $\sigma_{l,r}$ and reaches a point $p$ on the $+x$-axis intersects $\overline{op}$ 
in at least $\left|\left\lfloor\frac{l-r}{4}\right\rfloor-1\right|+1$
vertical segments of $C$ if $\sigma_1=\lt$;
otherwise, if $\sigma_1 = \rt$,
$\left\lfloor\frac{l-r}{4}\right\rfloor +1$ vertical segments of $C$.
\end{theorem}

\subsubsection{Lower bounds on \texorpdfstring{$x^+_{l,r}$}{x+}}
\label{subsubsec:lower_bound_x_minus}

By Theorem~\ref{thm:vertical-seg-crossing}, there are at least  
$\left|\left\lfloor \frac{l-r}{4}\right\rfloor-1\right|+1$ or $\left\lfloor \frac{l-r}{4}\right\rfloor+1$
vertical segments lying between $o$ and $p$, which implies that the distance from $o$ to $p$
is at least $\left|\left\lfloor \frac{l-r}{4}\right\rfloor-1\right|+2$ or $\left\lfloor \frac{l-r}{4}\right\rfloor+2$.  
Besides this term, we will show that for some turn sequence, 
as in the upper bound on $x^+_{l,r}$,
the length of the maximal staircases from $o$ or to $p$ is also contributed
to the lower bound on $x^+_{l,r}$.

\begin{figure}[t]
\centering
\includegraphics[]{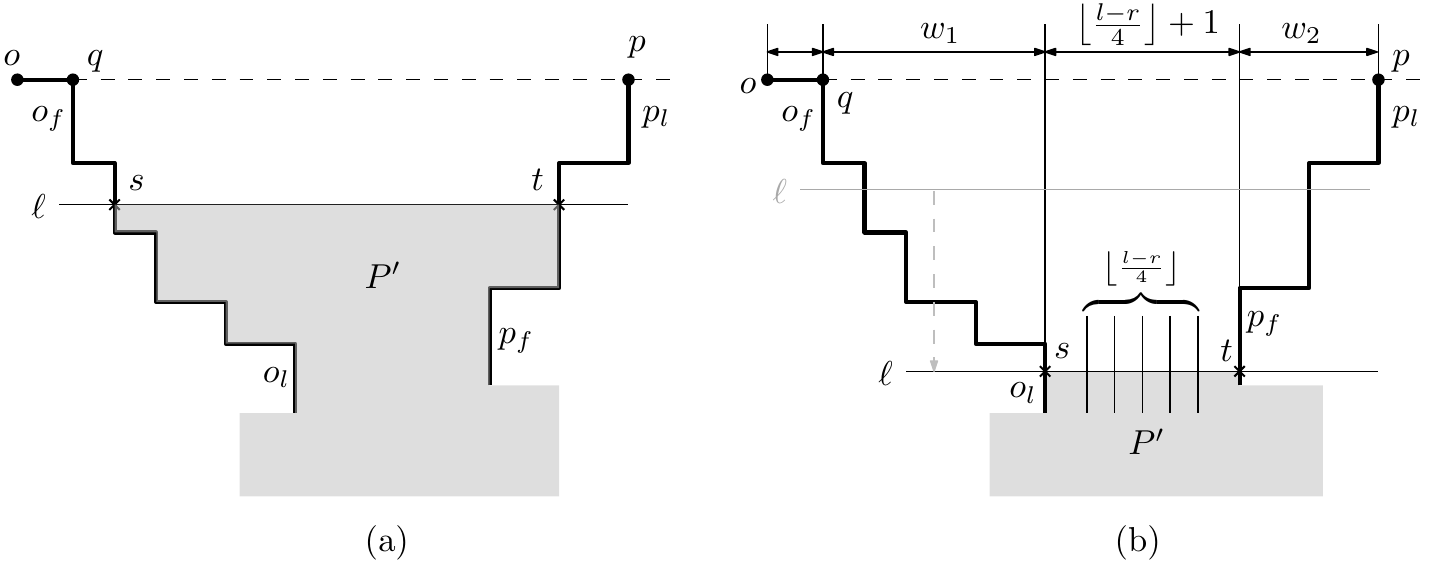}
\caption{
(a) New rectilinear polygon $P'$ bounded above by $\overline{st}$.
(b) Vertical segments between $s$ and $t$ when $\ell$ passes $o_l$ and $p_f$.
 }
\label{fig:tark7}
\end{figure}

\paragraph*{Case 1: $l-r \equiv 1 \pmod 4$}
\subparagraph*{Subcase 1.1: $\sigma_1=\rt$.}
It is obvious that $_{+x}^{}m_{-y}^{o} > 0$ and $_{+x}^{}m_{+y}^{p} > 0$.
For the maximal SE-staircase $C_o$ that starts from $o$,
let $o_f$ and $o_l$ denote the leftmost and rightmost vertical segments of $C$, respectively.
Similarly, for the maximal NE-staircase $C_p$ that ends at $p$,
let $p_f$ and $p_l$ denote the leftmost and rightmost vertical segments of $C$, respectively. 
See Figure~\ref{fig:tark7}(a).

We consider a horizontal line $\ell$ that intersects the interiors of two vertical
segments of $C_o$ and $C_p$ each at points $s\in C_o$ and $t\in C_p$.
Consider a rectilinear polygon $P'$ such that its boundary consists of 
the subchain $C[s,t]$ from $s$ to $t$ and a horizontal segment $\overline{ts}$.
The subchains $C[q,s]$ and $C[t,p]$ are staircases, and 
they have the same number of left and right turns.
This means that $R(P') = R(P)$, where $P$ is a rectilinear polygon
defined by connecting $\overline{pq}$ as in Section~\ref{subsubsec:verticalnumber}.
By Lemma~\ref{lem:rot and self}, there are at least 
$R(P')-1 = R(P)-1 = \left\lfloor\frac{l-r}{4}\right\rfloor$ vertical
segments that intersect $\overline{ts}$, so the length of $\overline{ts}$
is at least $\left\lfloor\frac{l-r}{4}\right\rfloor+1$. 
The distance between $o$ and $p$ is expressed as the sum
of four terms, as shown in Figure~\ref{fig:tark7}(b); 
(1) the length of $\overline{oq}$, which is at least one,
(2) the width $w_1$ of $C[q,s]$,
(3) the length of $\overline{st}$, and (4) the width $w_2$ of $C[t,p]$.
Their sum is at least $1 + w_1 + (\left\lfloor\frac{l-r}{4}\right\rfloor+1) + w_2$.

The value of $w_1+w_2$ gets bigger as $\ell$ moves lower 
because $w_1$ and $w_2$ are the width of the SE- and NE-staircases.
Note here that $\overline{st}$ always intersects at least 
$\left\lfloor\frac{l-r}{4}\right\rfloor$ vertical segments 
regardless of the position of $\ell$.
We also know that $\ell$ can intersect at least one of $o_l$ and $p_f$.
If $\ell$ intersects both $o_l$ and $p_f$, then $w_1+1 =\, _{+x}^{}m_{-y}^{o}$ 
and $w_2 =\, _{+x}^{}m_{+y}^{p}-1$, which is the maximum.
If $\ell$ intersects only $o_l$, then 
$w_1+w_2+1 \geq\, _{+x}^{}m_{-y}^{o}$.
Otherwise, if $\ell$ intersects only $p_f$, then 
$w_1+w_2+1 \geq\, _{+x}^{}m_{+y}^{p}$.
Thus we have that $w_1+w_2+1\geq \min\{_{+x}^{}m_{-y}^{o},\, _{+x}^{}m_{+y}^{p}\}$,
so 
$x^+_{l,r}\geq \min\{_{+x}^{}m_{-y}^{o},\,_{+x}^{}m_{+y}^{p}\}
+ \left\lfloor\frac{l-r}{4}\right\rfloor+1 = \min\{_{+x}^{}m_{-y}^{o},\,_{+x}^{}m_{+y}^{p}\}
+ \left\lfloor\frac{l-r+2}{4}\right\rfloor+1$
because $\left\lfloor\frac{l-r}{4}\right\rfloor = \left\lfloor\frac{l-r+2}{4}\right\rfloor$
for $l-r\equiv 1 \pmod 4$.

\subparagraph*{Subcase 1.2: $\sigma_1=\lt$.}
Unlike Case 1.1, the maximal staircase $C_o$ (containing $o$) and
the maximal staircase $C_p$ (containing $p$) are on the opposite sides of the $x$-axis,
so there is no horizontal line $\ell$ that intersects $C_o$ and $C_p$ at the same time.
This implies that the minimum number of vertical segments that cross $\overline{op}$,
which is $\left|\left\lfloor\frac{l-r}{4}\right\rfloor-1\right|+1 = \left|\left\lfloor\frac{l-r+2}{4}\right\rfloor-1\right|+1$,
is only a term of the lower bound. \\

For $l-r \equiv 1 \pmod 4$, we have the following bound.

\begin{lemma}
\label{lem:lowerbound_case1}
For any turn sequence $\sigma_{l,r}$ with $l-r \equiv 1 \pmod 4$, 
if $\sigma_1 = \rt$, then
$$x^+_{l,r}\geq \min\{_{+x}^{}m_{-y}^{o},\,_{+x}^{}m_{+y}^{p}\}
+ \left\lfloor\frac{l-r+2}{4}\right\rfloor+1,$$
if $\sigma_1=\lt$, then
$$x^+_{l,r} \geq\, \left|\left\lfloor\frac{l-r+2}{4}\right\rfloor-1\right|+2.$$
\end{lemma}

\paragraph*{Case 2: $l-r \equiv 3 \pmod 4$}
If $\sigma_1=\lt$, then the maximal staircases $C_o$ and $C_p$ are above the $x$-axis.
Moreover, two staircases are monotone to the $+x$-axis.
We can apply the same method used in Case 1 to bound $x^+_{l,r}$.
If $\sigma_1 = \rt$, then
two staircases are on the opposite sides of the $x$-axis, 
so only the minimum number of vertical segments that cross $\overline{op}$
determines the lower bound.

\begin{lemma}
\label{lem:lowerbound_case2}
For any turn sequence $\sigma_{l,r}$ with $l-r \equiv 3 \pmod 4$, if $\sigma_1 = \lt$, then
$$x^+_{l,r}\geq\min\{_{+x}^{}m_{+y}^{o},\, _{+x}^{}m_{-y}^{p}\} + 
\left|\left\lfloor\frac{l-r}{4}\right\rfloor-1\right|+1$$
$$= \min\{_{+x}^{}m_{+y}^{o},\,_{+x}^{}m_{-y}^{p}\} + 
\left|\left\lfloor\frac{l-r+2}{4}\right\rfloor-2\right|+1,$$
if $\sigma_1=\rt$, then
$$x^+_{l,r}\geq \left\lfloor\frac{l-r}{4}\right\rfloor+2
=\left\lfloor\frac{l-r+2}{4}\right\rfloor+1.$$
\end{lemma}

\paragraph*{Case 3: $l-r \equiv 2 \pmod 4$}
For this case, 
the maximal staircase $C_p$ is a SW-staircase or NW-staircase, i.e.,
goes to the west (to the $-x$-axis) 
while the maximal staircase $C_o$ goes to the east.
From this, we know that 
the lower bound is determined only by the minimum number of vertical segments that 
cross $\overline{op}$. 
We can bound $x^+_{l,r}$ as follows.

\begin{lemma}
\label{lem:lowerbound_case3}
For any turn sequence $\sigma_{l,r}$ with $l-r \equiv 2 \pmod 4$, if $\sigma_1 = \lt$, then
$$x^+_{l,r}\geq \left|\left\lfloor\frac{l-r}{4}\right\rfloor-1\right|+2
= \left|\left\lfloor\frac{l-r+2}{4}\right\rfloor-2\right|+2,$$
if $\sigma_1=\rt$, then
$$x^+_{l,r}\geq \left\lfloor\frac{l-r}{4}\right\rfloor+2
=\left\lfloor\frac{l-r+2}{4}\right\rfloor+1.$$
\end{lemma}

This lower bound is exactly matched with the upper bound in Lemma~\ref{lem:bounding-x-plus}
for any $l-r\equiv 2 \pmod 4$ except when $l-r=2$.

\paragraph*{Case 4: $l-r \equiv 0 \pmod 4$}
The lower bound for this case can be easily derived from the one for
the turn sequence
of $l-r\equiv 1 \pmod 4$ or $l-r\equiv 3 \pmod 4$, obtained by 
deleting the last turn from the original sequence. 
Using Lemma~\ref{lem:lowerbound_case1} and Lemma~\ref{lem:lowerbound_case2}
for the cases,
we can bound $x^+_{l,r}$ as follows.

\begin{lemma}
\label{lem:lowerbound_case4}
For any turn sequence $\sigma_{l,r}$ with $l-r \equiv 0 \pmod 4$ and $l-r \geq 8$,
if $\sigma_1 = \lt$, then
$$x^+_{l,r}\geq \min\{_{+x}^{}m_{+y}^{o},\, _{+x}^{}m_{-y}^{p},2\}
+\left\lfloor\frac{l-r+2}{4}\right\rfloor,$$
if $\sigma_1=\rt$, then
$$x^+_{l,r}\geq \left\lfloor\frac{l-r+2}{4}\right\rfloor +2.$$
For any turn sequence $\sigma_{l,r}$ with $l-r = 0$, 
if $\sigma_1 = \lt$, then
$x^+_{l,r}\geq \min\{_{+x}^{}m_{+y}^{o},\, _{+x}^{}m_{-y}^{p},2\}+2,$
otherwise,
$x^+_{l,r}\geq \min\{_{+x}^{}m_{-y}^{o},\, _{+x}^{}m_{+y}^{p},2\}+2.$
For any turn sequence $\sigma_{l,r}$ with $l-r = 4$,
$x^+_{l,r}\geq 3.$
\end{lemma}

We conclude this section with the following theorem.

\begin{theorem}
\label{thm:lowerboundconclusion}
For any turn sequence $\sigma_{l,r}$ with $l-r \geq 0$,
the lower bounds on $x^+_{l,r}$ can be summarized in
Lemma~\ref{lem:lowerbound_case1}, Lemma~\ref{lem:lowerbound_case2}, Lemma~\ref{lem:lowerbound_case3}, and Lemma~\ref{lem:lowerbound_case4}
for $l-r\equiv 1, 3, 2, 0 \pmod 4$, respectively.
\end{theorem}

\remark{
The lower bounds for $l-r=5$ and for $l-r \equiv 2 \pmod 4$ with $l-r\geq 6$ 
are exactly matched with their upper bounds.}

\subsubsection{Lower bounds on \texorpdfstring{$x^-_{l,r}$}{x-}}
By the method we used in Section~\ref{subsubsec:verticalnumber},
we can also bound the minimum number of vertical segments of $C$ that 
cross $\overline{op}$, where $p$ is on the $-x$-axis. 
The parameters that determine the bound are the last turn $\sigma_n$ 
(not the first turn $\sigma_1$) 
and the excess number $l-r$ only (not including the lengths of the maximal staircases).
We here summarize the results without giving the detailed proofs.

\begin{theorem}
\label{thm:lowerboundconclusion2}
\begingroup
For any turn sequence $\sigma_{l,r}$ with $l-r \geq 0$, the lower bounds on $x^-_{l,r}$ are as follows:
\renewcommand{\arraystretch}{1.3} 
\begin{center}
\begin{tabular}{|l|l|ll|l|} 
\hline
         & $l-r \equiv 0$~or~$1 \pmod 4$                                                        & \multicolumn{2}{l|}{$l-r \equiv 2 \pmod 4$}                                                    & $l-r \equiv 3 \pmod 4$                                         \\ 
\hline
\multirow{2}{*}{$x^-_{l,r}\geq$} & \multirow{2}{*}{$\left| \left\lfloor \frac{l-r+2}{4}\right\rfloor - 1 \right| + 1$} & $\left| \left\lfloor \frac{l-r+2}{4}\right\rfloor - 2 \right| + 2$ & if~$\sigma_n=\texttt{L}$ & \multirow{2}{*}{$\left\lfloor \frac{l-r+2}{4}\right\rfloor$}  \\ 
\cline{3-4}                &                                                                                     & $\left\lfloor \frac{l-r+2}{4}\right\rfloor+1$  & if~$\sigma_n=\texttt{R}$                    &                                                               \\
\hline
\end{tabular}
\end{center}
\endgroup
\end{theorem}

\remark{
It is worthwhile to mention that unlike the bounds on the $+x$-axis,
the lower bound on the $-x$-axis for any $l-r \geq 3$ is
exactly matched with the upper bound. 
The bound is also tight for $l-r = 0$, but not for $l-r = 1, 2$.
}

\section{Concluding remarks}

In this paper,
we characterize combinatorial and geometric properties on
the reachable region by a turn sequence of left and right turns.
For this, we first present the sufficient and necessary conditions on
the reachability to the signed axes. 
We next obtain upper bounds on the maximum distance to the closest
reachable point from the origin on the signed axes 
by describing drawing algorithms of the turn sequence, and
prove the lower bounds by bounding the number of self-intersections of 
a (non-simple) rectilinear polygon induced by the turn sequence,
which are almost tight within some additive constant for some signed axes.
Interestingly, these bounds are expressed in terms of the difference of the number
of left and right turns and the length of 
the maximal monotone prefix or suffix of the sequence. 

We close this section with a list of open problems. 
First, the upper and lower bounds for some cases are not tight; 
for example, for the sequence with $l-r=2$ and the last turn of $\lt$,
we have that $3\leq x^-_{l,r} \leq\, _{-x}^{}m_{+y}^{p}+2$; 
the bounds are not tight within an additive constant.
It remains open to narrow the gaps between the bounds or
find the exact closest reachable point in polynomial time.
Second, we can consider an interesting variant of characterizing
the reachable region by a \emph{quad-turn sequence}, 
as a natural extension of the binary-turn sequence seen so far, 
which is a sequence consists of two different left turns $\lt_1$ and $\lt_2$, 
and two different right turns $\rt_1$ and $\rt_2$ 
as shown in Figure~\ref{fig:tark8}(a).
For a quad-turn sequence $\tau = \tau_1\tau_2\cdots\tau_n$ 
where $\tau_i \in \{\lt_1,\lt_2,\rt_1,\rt_2\}$,
a chain realizing this sequence is drawn in a \emph{triangular grid}
in Figure~\ref{fig:tark8}(a).
This grid can be deformed to a \emph{right triangular grid} 
like Figure~\ref{fig:tark8}(b).
While the binary-turn sequence has 
the rotation number $\left\lfloor \frac{l-r}{4}\right\rfloor$ as
a term of the distance bound, the quad-turn sequence has
a term of $\left\lfloor \frac{l_1+2l_2-r_1-2r_2}{6}\right\rfloor$,
where $l_1, l_2, r_1$, and $r_2$ denote 
the number of $\lt_1, \lt_2, \rt_1$, and $\rt_2$ in $\tau$, respectively.
We can also consider the \emph{hexagonal chain} drawing 
in the triangular grid for turn sequences
that only contain $\lt_1$ and $\rt_1$ turns.

\begin{figure}[t]
\centering
\includegraphics[width=10cm]{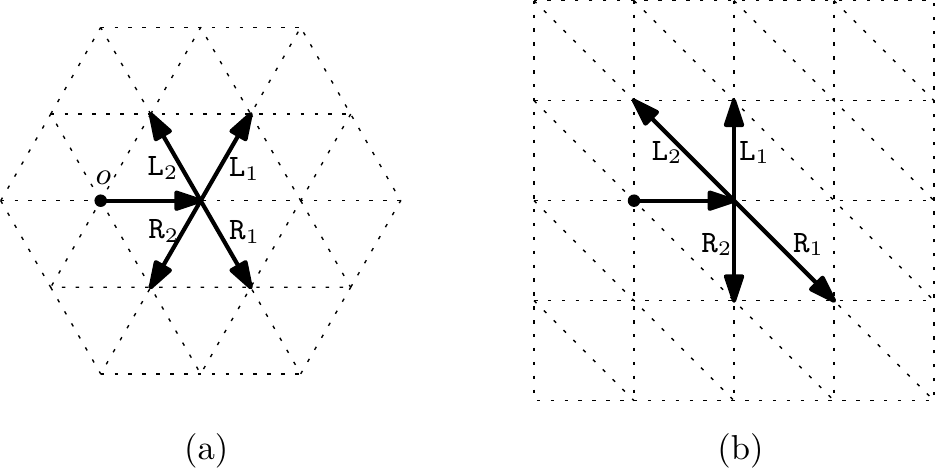}
\caption{
(a) Four possible turns on the triangular grid.
(b) Four possible turns on the right triangular grid.
 }
\label{fig:tark8}
\end{figure}

\bibliography{reachability}

\end{document}